\documentclass[a4paper,onecolumn,11pt,accepted=2026-02-09]{quantumarticle}
\pdfoutput=1
\usepackage[utf8]{inputenc}
\usepackage[english]{babel}
\usepackage[T1]{fontenc}
\usepackage{amsmath,amssymb}
\usepackage{hyperref}
\usepackage[numbers,sort&compress]{natbib}
\usepackage{tikz}
\usepackage{lipsum}

\usepackage{algorithm}
\usepackage{algpseudocode}
\usepackage{caption}
\usepackage{subcaption}
\usepackage{physics}
\usepackage{soul}
\usepackage{diagbox}
\usepackage{bm}
\usepackage{amsthm}
\usepackage{amsfonts}
\usepackage{fancyhdr}
\usepackage{color}
\usepackage{mathtools}
\usepackage{tikz}
\usepackage{lipsum}
\usepackage{appendix}
\usepackage{multirow}
\usepackage{thmtools}
\usepackage{thm-restate}

\usepackage[normalem]{ulem}

\newtheorem{theorem}{Theorem}
\newtheorem{lemma}{Lemma}

\newtheorem{definition}[lemma]{Definition}
\newtheorem{corollary}{Corollary}

\newcommand{\unit}{1\!\!1}

\newcommand{\phyM}{\mathcal{S}}

\newcommand{\qs}{\mathcal{Q}}
\newcommand{\ns}{\mathcal{NS}}
\newcommand{\poly}{\mathcal{P}}

\newcommand{\Var}[1]{\langle{#1}\rangle}
\newcommand{\prob}[1]{\operatorname{Pr}\left({#1}\right)}

\begin{document}
	
	\title{More entropy from shorter experiments using polytope approximations to the quantum set}
	
    \author{Hyejung H. Jee*} 
    \affiliation{Quantinuum, Partnership House, Carlisle Place, London SW1P 1BX, UK}
    \email{hailey.jee@quantinuum.com}
    \author{Florian J. Curchod}
    \affiliation{Quantinuum, Terrington House, 13–15 Hills Road, Cambridge CB2 1NL, UK}
    \author{Mafalda L. Almeida**}
    \affiliation{Quantinuum, Terrington House, 13–15 Hills Road, Cambridge CB2 1NL, UK}
    \email{mafalda.almeida@quantinuum.com}

	\maketitle
	\thispagestyle{fancy}
	
	\begin{abstract}

We introduce a systematic method for constructing polytope approximations to the quantum set in a variety of device-independent quantum random number generation (DI-QRNG) protocols. Our approach relies on two general-purpose algorithms that iteratively refine an initial outer-polytope approximation, guided by typical device behaviour and cryptographic intuition. These refinements strike a balance between computational tractability and approximation effectiveness. By integrating these approximations into the probability estimation (PE) framework [Zhang et al., PRA 2018], we obtain significantly improved certified entropy bounds in the finite-size regime. We test our method on various bipartite and tripartite DI-QRNG protocols, using both simulated and experimental data. In all cases, it yields notably higher entropy rates with fewer device uses than the existing techniques. We further extend our analysis to the more demanding task of randomness amplification, demonstrating major performance gains without added complexity. These results offer an effective and ready-to-use method to prove security---with improved certified entropy rates---in the most common practical DI-QRNG protocols. Our algorithms and entropy certification with PE tools are publicly available under a non-commercial license at \url{https://github.com/CQCL/PE_polytope_approximation}.
	\end{abstract}

	\section{Introduction}
	\label{sec:intro}

Device-independent randomness certification is among the key achievements of quantum cryptography. By making security claims that rely minimally on hardware characterisation, this approach enables random number generators that are secure against computationally unbounded adversaries, while minimising the risk of mismatches between the theoretical model and practical implementation.
This is achieved by repeatedly challenging a quantum device and analysing its responses in what is called a multi-prover interactive proof system. The price for this uplift in security is (i) having to build more complex, expensive and/or slower devices and (ii) observing a fall in the rate of entropy generation. As a consequence, the design of DI-QRNG protocols should pay particular attention to maximising the amount of certified randomness per device use. A related requirement is to minimise the number of device uses needed to obtain a positive rate of certified entropy. Indeed, for most applications (imperfect) raw output entropy goes through randomness extraction algorithms which are demanding in computational resources, both in time and memory. In order to avoid this post-processing bottleneck\footnote{For an excellent discussion of this bottleneck refer to, e.g., Appendix E.C of \cite{hayashi2016more} or the concrete examples in Section 5 of \cite{foreman2025cryptomite}.}---that easily becomes a roadblock in practical applications---one needs to minimise the input size of these algorithms and thus minimise the number of device uses, $n$, in the protocol. 

Standard techniques for lower-bounding the output entropy in DI protocols include the method described in \cite{pironio2010random,pironio2013security,fehr2013security} for quantum adversaries classically correlated with the device and the entropy accumulation theorem (EAT) introduced in \cite{Dupuis:2020EAT} for quantum side information. Note that although classical side information is a well-justified assumption in QRNG protocols (see e.g.~\cite{pironio2013security}), assuming quantum side information is necessary for most quantum key distribution (QKD) protocols. While these techniques are proven to be optimal in the asymptotic limit of very large $n$, their performance deteriorates significantly in the finite-size regime compatible with real-world applications. This limitation is addressed by the probability estimation (PE) framework for classical side information~\cite{zhang2018certifying,knill2020generation}, which derives entropy bounds optimised for a given finite $n$.

Despite considerable success in improving the finite-size performance of several DI-QRNG implementations \cite{zhang2018certifying,knill2020generation,shalmPEFinblocks}, a limitation of the PE technique is that it requires an optimisation to be performed over sets of allowed behaviours that must be described by convex polytopes. Allowed behaviours are any device input-output probability distributions compatible with a quantum realisation (through the Born rule) and, in the case of DI-QRNG protocols proposed so far, with the device being composed of separated non-communicating components where measurements take place. Given that quantum theory imposes non-linear constraints in this setup, any practical applications of the PE framework require polytope approximations that contain the (non-polytope) convex set of quantum behaviours. 

Finding good polytope approximations, however, remains a significant challenge. On one hand, coarse outer-approximations to the quantum set tend to greatly overestimate the strategies available to an adversary, leading to unnecessarily weak entropy bounds. On the other hand, achieving fine-grained approximations in the high-dimensional spaces that characterise these sets quickly becomes computationally intractable. So far, only basic approximations in the simplest Bell scenario\footnote{The so-called CHSH scenario, consisting of two non-communicating devices with two measurement choices and outputs each.} have been proposed: in \cite{zhang2018certifying,knill2020generation}, the authors consider the polytope defined by the subset of non-signalling behaviours satisfying Tsirelson's bound for CHSH (i.e., that do not exceed its maximum quantum violation); while in \cite{bierhorst2020tsirelson} the authors go a step further by excluding from this polytope all the behaviours violating a chosen asymmetric CHSH inequality \cite{acin2012randomness} above its quantum bound. 

\subsection{Contributions}\label{sec:contributions}

Our main contribution is to introduce a systematic method for constructing polytope approximations to the quantum set across a variety of DI-QRNG setups. To this end, we propose two general-purpose algorithms that balance computational tractability and approximation effectiveness. The key lies in designing the polytopes based both on the typical behaviour of the device, and on educated guesses of possible adversarial strategies.

More precisely, our approach begins with any initial polytope that contains the quantum set and proceeds iteratively by removing exclusively non-quantum regions of this polytope. As a result, we construct a sequence of polytopes that approximate the quantum set with increasing accuracy. Since the complexity of computing successive approximations grows rapidly, our algorithms carefully select the non-quantum behaviours that significantly contribute to the adversary's guessing power. This is achieved by identifying suitable quantum Bell inequalities according to the typical behaviour of the device and cryptographic intuition about potential adversarial strategies. We identify two effective strategies, each giving rise to a different polytope approximation algorithm.

We illustrate the performance of our approach by integrating it with the PE framework to derive new certified entropy bounds for a variety of practical DI-QRNG protocols. For this purpose, we consider setups based on bipartite and tripartite Bell tests involving two dichotomic observables per party. We analyse both theoretical noisy behaviours that resemble those that could be obtained in experimental implementations, as well as data from real-world loophole-free CHSH-Bell tests. In all cases studied, our method yields significantly higher entropy bounds with fewer device uses compared to previous analyses. 

In addition to standard protocols where the inputs are independent of the devices (e.g., randomness expansion or QKD), we apply our method to a bipartite scenario in which the device's internal state may be correlated with the inputs, typically referred to as randomness amplification protocols. These are particularly challenging to analyse, which often leads to suboptimal entropy rates. Nonetheless, we demonstrate through examples that---with no more effort than required for protocols assuming independent sources---our method achieves improvements of several orders of magnitude over existing alternatives.\\

Our construction can readily be applied to any practical DI-QRNG (based on bi- or tripartite Bell tests with two inputs and outputs per party), providing optimised entropy rates in the finite $n$ regime. To this end, we provide a software implementation (in Python) of both our algorithms for polytope approximations and to obtain entropy bounds using the PE framework. The codes are available at \url{https://github.com/CQCL/PE_polytope_approximation}.

Although we have not tested the effectiveness of the polytope approximations for protocols with more parties, our ideas extend to these more complex scenarios. A similar approach could be considered for semi-DI-QRNGs where the set of allowed behaviours can be bounded using semidefinite constraints (see \cite{Tavakoli_SDP_review} for a survey).

This paper is structured as follows. In section~\ref{sec:PEframework} we provide a concise overview of DI entropy certification protocols and the PE framework, along with a description of how this can be applied in practice. In Section~\ref{sec:polytope}, we describe our algorithms for constructing polytope approximations to the set of quantum correlations. This section is self-contained and may be of independent interest. Section~\ref{sec:DIrandVerif} showcases our technique across a variety of DI-QRNG scenarios. Finally, we present our conclusions in Section~\ref{sec:conclusion}.
	
	\section{The PE framework}
	\label{sec:PEframework}

	The probability estimation framework for randomness certification with classical side information was introduced by Zhang et al.~\cite{zhang2018certifying} and a more thorough description of the framework can be found in the follow-up paper \cite{knill2020generation}. The goal of this section is to provide a summary of the essential steps of the construction along with some simplified proofs.

	\subsection{Protocol basics and notation}
	
Consider a user who can query a physical device that takes inputs represented by a random variable $Z$ and returns outputs $C$. Over an $n$-round protocol, the sequences of inputs and outputs are described by $\mathbf{Z}=Z_1\ldots Z_n$ and $\mathbf{C}=C_1\ldots C_n$, respectively. The user is interested in lower-bounding the uncertainty in $\mathbf{C}$ from the perspective of a computationally unbounded adversary (`Eve'), having access to a complete description of the device and holding classical correlations $E$ with it. The random variables $\mathbf{CZ}E$ of the process follow some  distribution $\mu$. Although the adversary's side information could be represented by the sequence $E_1, \dots, E_n$, here we consider that all this information is merged into a single random variable $E$ available to the adversary before the protocol starts. While in most cases this implies an overestimation of Eve's information, it allows the $n$-round behaviour of the device to be conditioned on a fixed realisation $E=e$, which considerably simplifies the security analysis. 

Inputs $\mathbf{Z}$ may be public, meaning that they can be known to Eve once generated, and follow a distribution $\mu(\mathbf{Z}|e)$ known to Eve. Inputs $Z_i$ are not required to be independent amongst themselves or even from the quantum device, given that they satisfy the Markov condition:
\begin{equation}\label{eq:condindinputs}
	\mu(Z_{i+1}|\mathbf{C}_{\leq i},\mathbf{Z}_{\leq i},e)=\mu(Z_{i+1}|\mathbf{Z}_{\leq i},e)\,,\qquad \forall i\,.
\end{equation}
This is essential to prevent information on past outputs $\mathbf{C}_{\leq i}=C_1\ldots C_i$ from being leaked through inputs $Z_{i+1}$ for any $e$.  

To certify entropy in the demanding DI scenario where the user sees the device as a black-box---i.e., can only rely on its input-output statistics---the setup must incorporate physical constraints that limit the set of allowed behaviours $\mu(\mathbf{C},\mathbf{Z}|e)$. In this context, we consider the standard scenario of a physical device composed of separate, non-communicating sub-units, whose non-communicating nature can be verified without accessing their internal mechanisms. This structure enables the user to identify a set $\phyM$ of allowed behaviours $\mu(C_iZ_i|e)$ at each round $i$.
Take for instance the simplest protocols considered here, based on a bipartite Bell scenario where the physical device is composed of two non-communicating units, $\mathsf{A}$ and $\mathsf{B}$, with inputs $Z=XY$ and outputs $C=AB$. As a result, any allowed distributions $\mu(A_iB_iX_iY_i|e)$ must be non-signalling~\cite{popescu1994quantum}, i.e. $\mu(A_i|X_iY_ie)=\mu(A_i|X_ie)$ and $\mu(B_i|X_iY_ie)=\mu(B_i|Y_ie)$.

Given our focus on practical applications, the set of allowed behaviours is further constrained by quantum theory---specifically, any $\mu(C,Z|e)$ must be compatible with a quantum realisation. This means that conditional behaviours $\mu(C|Z,e)$ must be describable by Born's rule; that is, $\mu(C|Z,e) = \tr(\rho M_C^Z)$, where the quantum state $\rho$ is a positive semi-definite Hermitian operator with unit trace, acting on the Hilbert space $\mathcal{H}$ of the system. The measurement operators $M_C^Z$ may be taken as projectors on $\mathcal{H}$, satisfying $\sum_c M_c^Z = \unit$. In the bipartite example---comprising two non-communicating sub-units---the measurement operators take the form of operators $M_A^X \otimes M_B^Y$ acting on $\mathcal{H}_\mathsf{A}\otimes\mathcal{H}_\mathsf{B}$.

The PE framework allows us to lower-bound the  entropy of $\mathbf{C}=C_1\ldots C_n$, or more generally any sequence $\mathbf{D}=D_1\ldots D_n$ where each $D_i$ is determined by the raw output $C_i$~\footnote{For instance, in the bipartite Bell scenario we could be interested in the output entropy of a single party, say $\mathsf{A}$, implying setting $D=A$.}. In particular, it shows how to design entropy witnesses $W(\mathbf{c},\mathbf{z})$ that are used to check if $\mathbf{D}$ has entropy above a certain threshold (see Figure~\ref{fig:scheme}), from the point of view of an adversary knowing the inputs $\mathbf{z}$ and the side information $e$. In this scheme, both the witness and its threshold are defined prior to the entropy certification stage and should be optimised for the \emph{typical behaviour} of the device. This behaviour can be estimated by querying the device over many rounds during a characterisation step taking place before the entropy certification stage (the latter is depicted in Figure~\ref{fig:scheme}). If desired, the witness can be redesigned in subsequent runs taking into account the observed behaviour over previous runs.  Notice that if the behaviour of the device deviates significantly from the previously estimated typical behaviour, the protocol may abort with greater likelihood, but its security is never compromised.

\begin{figure}[t]
    \centering
    \includegraphics[width=\textwidth]{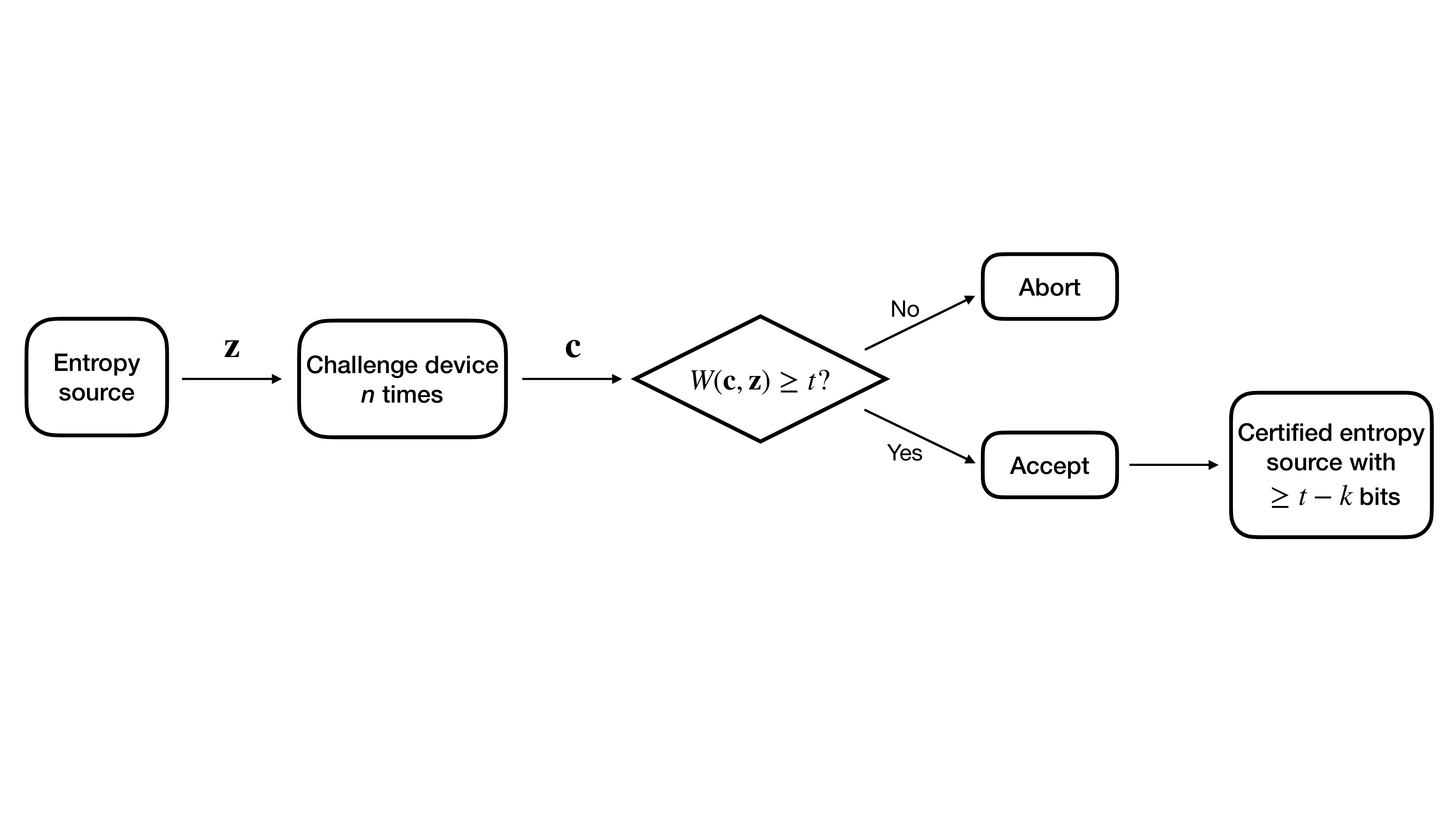}
    \caption{\textbf{Schematic of DI randomness certification protocol}. An entropy source $\mathbf{Z}$ is used to sequentially challenge $n$ times a quantum device, composed of non-communicating sub-units, which for each challenge $z_i$ return an answer $c_i$. The data $\mathbf{(c,z)}$ is evaluated by a pre-defined linear entropy witness $W$. If an entropy larger than a (pre-defined) threshold $t$ is witnessed, the protocol produces a certified entropy source with at least $t-k$ bits. The constant $k$ is typically a small number that depends on the witness design and on a chosen security parameter $\varepsilon$ (see for instance \eqref{eq:extractsmoothH}).}
    \label{fig:scheme}
\end{figure}

\paragraph{Basic entropy measures.}
In this paper, we will use measures of entropy tailored for our specific setup that build on a few basic definitions. 
Throughout, we adopt the convention where entropies are measured in bits; accordingly all logarithms are taken to base 2.

The first entropy measure is the \emph{min-entropy}, which captures the minimum uncertainty of a random variable $X$ distributed according to a probability distribution $\mu$. It is defined as 
\begin{equation}\label{def:Hmin}
	H_{\min} (X)=-\log P_{\rm guess}(X)\,,
\end{equation}
where 
\begin{equation}\label{def:pguess}
	P_{\rm guess}(X):=\max_x \mu(X=x)
\end{equation}
is the \emph{guessing probability} of $X$. When necessary, we will explicitly write $H_{\min ,\mu}(X)$ to indicate the underlying distribution. The \emph{conditional min-entropy} of $X$ given a fixed value $Y=y$ (hereafter denoted simply by $y$) is defined as 
\begin{equation}\label{eq:condHmin}
	H_{\min}(X|y)=-\log \max_x \mu(x|y)\,.
\end{equation}
The second measure we consider is the \emph{Shannon entropy}, which quantifies the average uncertainty of $X$,
\begin{equation}
	H(X)=\mathbb{E}(-\log \mu(X))\,,
\end{equation}
where the expectation operator is given by $\mathbb{E}(Y)=\sum_y\mu(y)y$. When needed, we will explicitly write $\mathbb{E}_\mu$ as a reminder of the distribution being considered.

\subsection{Probability estimation factors (PEFs)}\label{sec:PEFs}
Let us start by introducing the building block of the probability estimation technique\footnote{In~\cite{knill2020generation}, a more general building block, known as the \emph{soft}-PEF, is introduced. This generalization differs from the standard PEF when $\mathbf{C} \neq \mathbf{D}$. However, we do not include this extension here, as it is more cumbersome to work with and unnecessary for our purposes---particularly since every PEF is also a soft-PEF (see Lemma 25 in~\cite{knill2020generation}).}.
\begin{definition}\label{def:PEF}
		Let $\beta>0$, and let $\phyM$ be the set of allowed behaviours. A \emph{probability estimation factor (PEF)} with power $\beta$ for $D|Z$ is a random variable $F=F(C,Z)$ such that 
	\begin{enumerate}
		\item  $F(C,Z)\geq 0$, \label{def:Fpos}
		\item  $\mathbb{E}\left(F(C,Z)\mu(D|Z)^\beta\right)\leq1\,,\qquad \forall \mu(C,Z) \in \mathcal{S}\,. $\label{def:PEFcons}
	\end{enumerate} 
	\end{definition}
Notice that $F$ can be evaluated using only the information available to the user. 

\begin{corollary}\label{th:PEFlargerbeta}
    If $F$ is a PEF with power $\beta$, then it is necessarily a PEF for any power $\beta'\geq \beta$. 
\end{corollary}

A key property of PEFs is that a product of sequential PEFs defines a new PEF, as we will see in the following theorem. This result requires the side information $e$ to be fixed throughout the process, so, for clarity, we explicitly include this dependency on $e$:
	
\begin{theorem}[Theorem 23 in \cite{knill2020generation}]\label{th:chainedPEF}
	Let $\phyM^n$ be the set of allowed behaviours in a sequential $n$-round protocol. Define the product of single round PEFs with power $\beta$ as $T_n=\Pi_{i=1}^nF_i$. If the Markov condition \eqref{eq:condindinputs} holds, $T_n$ is a PEF with power $\beta$ for $\mathbf{D}|\mathbf{Z}$, i.e.  
	\begin{enumerate}
		\item $T_n(\mathbf{C},\mathbf{Z})\geq 0$, \label{eq:Tnpos}
		\item $\mathbb{E}\left(T_n(\mathbf{C},\mathbf{Z})\mu(\mathbf{D}|\mathbf{Z},e)^\beta\right)\leq1\,,\qquad \forall \mu(\mathbf{C,\mathbf{Z}}|e) \in \phyM^n$.\label{eq:Tnconst}
	\end{enumerate} 
\end{theorem}
\noindent Note that according to Corollary~\ref{th:PEFlargerbeta}, we can compose sequences of individual PEFs with different powers $\beta_i$ by choosing $\beta=\max_i \beta_i$. 

Using Theorem~\ref{th:chainedPEF} and Markov's inequality, we can now obtain a high-confidence lower-bound on the conditional probability of $\mathbf{D}$:
	\begin{equation}
		\prob{T_n(\mathbf{C},\mathbf{Z}) \mu(\mathbf{D}|\mathbf{Z},e)^\beta\geq \frac{1}{\kappa}}\leq \kappa  \mathbb{E}\left(T_n(\mathbf{C},\mathbf{Z}) \mu(\mathbf{D}|\mathbf{Z},e)^\beta\right)\leq \kappa\,,
	\end{equation}
where $ \kappa > 0$ is a chosen security parameter. This implies a confidence bound on any $\mu(\mathbf{D}|\mathbf{Z},e)$ that belongs to $\phyM^n$:
	\begin{equation}\label{eq:confboundp}
		\prob{  -\log\mu (\mathbf{D}|\mathbf{Z},e)\geq \frac{\log T_n(\mathbf{C},\mathbf{Z})}{\beta}-\frac{\log{1/\kappa}}{\beta}}\geq 1-\kappa\,,\quad \forall e.
	\end{equation}
Given that $\log T_n(\mathbf{C,\mathbf{Z}})=\sum_{i=1}^{n}\log F_i(C_i,Z_i)$, this multi-round bound reduces to a function of the single round PEFs and the confidence level $\kappa$.

\subsection{From PEFs to extractable entropy}\label{sec:ExtractEntropy}
	
Here we will see how the confidence bound \eqref{eq:confboundp} allows us to establish a lower bound on the randomness of $\mathbf{D}$. To identify the relevant measure of entropy in our context, recall that the goal of a QRNG protocol is to output a string of near-perfect random bits independent from any other variables. This can be achieved by processing $\mathbf{D}$ through a randomness extraction algorithm \cite{boundedStorageKonigTerhal}. The length of the near-perfect random bit string that is generated as a result will depend on the chosen algorithm and on the amount of \emph{extractable entropy} in $\mathbf{D}$. The latter can be quantified by smooth min-entropies, which were proposed by \cite{RenWol04a} for the general case of adversaries holding quantum side information. Since we restrict our attention to adversaries classically correlated to the device, we will now introduce a lower bound on the (classical) 
smooth conditional min-entropy and from there, find a bound on the extractable entropy of $\mathbf{D}$.

\begin{definition}\label{def:smoothHmin}
		Let random variables $\mathbf{DZ}E$ be jointly distributed with $\mu$.  $\mathbf{D}$ has \emph{$\epsilon$-smooth $\mathbf{Z}$E-conditional min-entropy} $H_{\min}^{\epsilon}(\mathbf{D}|\mathbf{Z},E)\geq h$ if there exists a distribution $\nu$, defined on the same domain as $\mu$, such that
		\begin{enumerate}
			\item $\|\mu-\nu\|\leq \epsilon$\,,
			\item $H_{\min,\nu}(\mathbf{D}|z,e)\geq h\,, \qquad \forall z, e$\,,
			
		\end{enumerate}
		where $\|\mu(X)-\nu(X)\|=\frac{1}{2}\sum_x|\mu(x)-\nu(x)|$ is the \emph{total variation distance}.
	\end{definition}

	The next theorem establishes a relationship between the confidence bound \eqref{eq:confboundp} and a lower bound on the smooth conditional min-entropy. 
	
	\begin{theorem}\label{th:confbound2minent}
		If, for all $e$, $\prob{  -\log \mu(\mathbf{D}|\mathbf{Z},e)\geq h}\geq 1-\epsilon$ then $H_{\min}^{\epsilon}(\mathbf{D}|\mathbf{Z},E)\geq h$.
	\end{theorem}
	\begin{proof}
		See Appendix~\ref{app:boundsHminproofs}.
	\end{proof}
	
	To identify the extractable entropy in our protocols, recall that they include an acceptance step, where the sequence of $n$ inputs and outputs is evaluated by some entropy estimator $W$. In our case, it is defined as
	\begin{equation}\label{eq:entropy_witness}
		W:=\frac{1}{\beta}\sum_i \log{F_i}\,.
	\end{equation}
Having chosen a set of PEFs $\{F_i\}_i$ with power $\beta$ and an acceptance threshold $t$, the output sequence $\mathbf{d}$ is accepted and proceeds to extraction if 
$W(\mathbf{c},\mathbf{z})\geq t$; otherwise, the protocol aborts (Figure~\ref{fig:scheme}). The following result finally establishes a lower bound on the extractable entropy of \textbf{D}.

	\begin{theorem}
		\label{thm:extractable_E}
		Consider the event $\verb|Accept|=\{W\geq t\}$ and its probability $p_{Acc}=\Pr(W\geq t)$. If, for all $e$, $\prob{-\log\mu(\mathbf{D}|\mathbf{Z},e)\geq W}\geq 1-\kappa$ and $p_{Acc}=\Pr(W\geq t),$ then the \emph{extractable entropy} of $\mathbf{D}$, defined by $H_{\min}^{\varepsilon/p_{Acc}}(\mathbf{D}|\mathbf{Z},E,\verb|Accept|)$, is lower-bounded by
		\begin{equation}\label{eq:extractsmoothH}
			H_{\min}^{\varepsilon/p_{Acc}}(\mathbf{D}|\mathbf{Z},E,\verb|Accept|)\geq t + \frac{\log \kappa}{\beta}+ \log (\varepsilon-\kappa) \,,
		\end{equation}
		where $\varepsilon>\kappa$ is the chosen security parameter.
	\end{theorem}
	
	\begin{proof}
		See Appendix~\ref{app:boundsHminproofs}. 
	\end{proof}
	
	Note that the threshold $t$ can be chosen according to  
	\begin{equation}\label{eq:accept_threshold}
		t=\mathbb{E}_{\mu_{\rm typ}}(W)-\delta\,,
	\end{equation}
    where the expected value of $W$ is evaluated on the typical behaviour of the device $\mu_{\rm typ}$ (see discussion in Section~\ref{sec:PEframework}), and $\delta$ is a parameter that allows us to tune the acceptance rate of the protocol. The extractable entropy \eqref{eq:extractsmoothH} then takes the form 
	\begin{equation}\label{eq:extEnt}
		H_{\min}^{\varepsilon/p_{Acc}}(\mathbf{D}|\mathbf{Z},E,\verb|Accept|)\geq \sum_i \frac{\mathbb{E}_{\mu_{\rm typ}} [\log F_i]}{\beta} -\frac{1}{\beta}\log \frac{1}{\kappa}-\log \frac{1}{\varepsilon-\kappa}-\delta \,,
	\end{equation} 
	where $\sum_i \mathbb{E}_{\mu_{\rm typ}}\left(\frac{\log F_i}{\beta}\right)$ represents the asymptotic entropy for very large $n$, while the remaining terms do not depend on $n$. Although practical protocols typically use $\delta>0$, hereafter we set $\delta=0$ for simplicity.
The PE approach then offers a trade-off between the asymptotic entropy term and the finite size/non-IID (independent and identically-distributed) penalty term $(\log 1/\kappa)/\beta$, that is known to be advantageous for small $n$~\cite{zhang2018certifying,knill2020generation}. With other standard approaches, such as  \cite{pironio2013security}  for classical and  \cite{Dupuis:2020EAT} for  quantum side information, the extractable entropy contains a penalty term that scales with $\sqrt{n}$ (see Appendix~\ref{app:extractEnt_other}), which becomes negligible for very large $n$ but is significant for lower numbers of rounds.

	\subsection{Finding good PEFs}\label{sec:findPEF}
	
At this point, we aim at designing PEFs that maximise the amount of extractable entropy \eqref{eq:extEnt} in our protocols for a typical behaviour $\mu_{\rm typ}$. This means finding the set of $\{F_i\}_i$ with power $\beta$ that maximises
\begin{equation}\label{eq:optPEF}
\max_{\beta, \{F_i\}_i,\kappa} \mathbb{E}_{\mu_{\rm typ}}(W)-\frac{1}{\beta}\left(\log \frac{1}{\kappa} \right)-\log\frac{1}{\varepsilon-\kappa}\,,
\end{equation}
where $\mathbb{E}_{\mu_{\rm typ}}(W)=\sum_i \mathbb{E}_{\mu_{\rm typ}}(\log F_i)/\beta$ according to \eqref{eq:entropy_witness}, $\varepsilon$ is the chosen security parameter, and $\kappa$ is a free optimisation parameter such that $\varepsilon>\kappa>0$. Note that the user has no means of finding a meaningful estimate of $\mu_{\rm typ}(C_i,Z_i|e)$ at every single round.
A reasonable approach is to optimise $W$ for data collected over many rounds and perform an estimate of $\mu_{\rm typ}$ assuming an IID behaviour of the device,  i.e. where inputs (outputs) are described by the same random variable $Z_i=Z$ ($C_i=C$) at every round. This entropy estimator will be suboptimal for the real, non-IID, experimental conditions but is, by definition, a valid estimator for any $\mu$ in the allowed set of behaviours.

Performing an optimisation for the IID behaviour greatly simplifies our problem. The probability distributions at every round can now be represented by a single distribution
    \begin{equation}
        p(C,Z)=\mu_{\rm typ}(C_i,Z_i|e)\,,
    \end{equation}
while the set of PEFs is reduced to a single function $F=F_i$. Consequently, the entropy witness \eqref{eq:entropy_witness} now becomes 
\begin{equation}\label{eq:W_IID}
    W=\frac{n}{\beta}\log F\,.
\end{equation}
As a result, optimisation \eqref{eq:optPEF} is performed over a much simpler set, and the expected value in the objective function simplifies to
	\begin{equation}\label{eq:entropy_estIID}
		\mathbb{E}_{\mu_{\rm typ}}(W)=n\frac{\mathbb{E}_p(\log F)}{\beta}\,.
	\end{equation} 
Distribution $p$ can be estimated by the frequencies of data  collected either in a characterisation stage of the protocol or in previous certification stages (see Section~\ref{sec:intro}), according to
\begin{equation}\label{eq:freq}
		\hat p(c,z)=\frac{1}{n}\sum_{i=1}^n\chi(c_i=c,z_i=z)\,,
	\end{equation}
	where $\chi$  is an indicator function such that $\chi(\verb|True|)=1$ and $\chi(\verb|False|)=0$. 
    
Testing the entropy of data $(\mathbf{c},\mathbf{z})$ collected during a randomness certification stage  (as depicted in Figure~\ref{fig:scheme}) corresponds to computing 
\begin{equation}W(\mathbf{c},\mathbf{z})=n\sum_{c,z} p_{\rm obs}(c,z)\frac{\log F(c,z)}{\beta}\,,
\end{equation}
where $p_{\rm obs}(c,z)$ are the frequencies in $(\mathbf{c},\mathbf{z})$ also using \eqref{eq:freq}.

	Finally, we can express the bound \eqref{eq:extEnt} on the extractable entropy in the more recognisable form, 
		\begin{equation}\label{eq:extEntIID}
		H_{\min}^{\varepsilon/p_{Acc}}(\mathbf{D}|\mathbf{Z},E,\verb|Accept|)\geq nt' -\frac{1}{\beta}\left(\log \frac{1}{\kappa} \right)-\log\frac{1}{\varepsilon-\kappa}\ \,,
	\end{equation} 
    where the asymptotically dominating term is the amount of entropy defined by the acceptance threshold $t$---which can also be interpreted as the product of
    the number of rounds $n$ by the single round entropy $t'=\mathbb{E}_p\left(\frac{\log F}{\beta}\right)$---and the penalty terms are due to memory and finite size statistics. The single round entropy estimator $(\log F)/\beta$ is a suboptimal estimator for the average Shannon entropy (refer to \cite{knill2020generation} and also Appendix \ref{app:PEFmintradeoff}). Moreover, in the asymptotic limit of very large $n$, it is proven that this estimator converges to optimality~\cite{knill2020generation}.

	\paragraph{Tackling the optimisation problem.} In order to find the optimal $W$ for an IID behaviour, we now have to solve the optimisation problem described in \eqref{eq:optPEF}:
	\begin{equation}\label{eq:fullPEFopt}
		\begin{aligned}
			\max_{F,\beta,\kappa} \quad & \frac{1}{\beta}\left(n\sum_{c,z} p(c,z)\log F(c,z)+\log (\kappa) \right)+\log(\varepsilon-\kappa)\\
			\textrm{s.t.} \quad & F\geq 0\\
			&\beta>0\\
		&\varepsilon>\kappa>0\\	&\sum_{c,z}\mu(c,z|e)F(c,z)\mu(d|z,e)^\beta\leq 1\,,\quad \forall \mu(C,Z|e) \in \phyM\,.\\
		\end{aligned}
	\end{equation}

\noindent 
Due to the non-linear objective and constraint, directly solving the full optimisation is intractable, and we therefore consider a relaxation of the above problem.
% Since no efficient method is known to directly solve the full optimisation problem, we instead consider a relaxation. 
Specifically, for fixed values of $(\beta, \kappa)$, we optimise over $F$ by solving:
	\begin{equation}
		\begin{aligned} \label{eq:PEFoptnobeta}
			\max_F\quad & \sum_{c,z} p(c,z)\log F(c,z)\\
			\textrm{s.t.} \quad & F \geq 0\,,\\	&\sum_{c,z}\mu(c,z|e)F(c,z)\mu(d|z,e)^\beta\leq 1\,,\quad \forall \mu(C,Z|e) \in \phyM\,.\\
		\end{aligned}
	\end{equation} 
We perform optimisation for a large number of $(\beta, \kappa)$ pairs to generate a collection of feasible solutions to the original problem \eqref{eq:fullPEFopt}. Each such solution provides a lower bound on the true optimum. Finally, we select the tuple $(F, \beta, \kappa)$ from this set that maximises the objective in \eqref{eq:fullPEFopt}. Although this strategy brings us closer to finding a good entropy estimator \eqref{eq:entropy_estIID} for the typical device behaviour, whether we can actually have a practical method to tackle the optimisation \eqref{eq:PEFoptnobeta} depends on the characterisation of $\phyM$.  

In this work, we consider two distinct device-independent scenarios. In the first, the inputs $\mathbf{Z}$ are chosen independently of the device's behaviour and therefore also independently of the adversary: $\mu(Z|e)=\mu(Z)$. In the second scenario, $\mathbf{Z}$ may share correlations with the device, restricted by a parameter known to the user. Each setting defines a different and characteristic set of allowed behaviours $\phyM$. 
However, as we will discuss in detail in Section \ref{sec:DIrandAmp}, both scenarios ultimately require an optimisation over conditional quantum behaviours $\mu(C|Z)$. We begin by analysing the first, simpler scenario, which provides the foundational tools necessary to address the more general case.

We then turn our attention to the set of conditional quantum behaviours, denoted by $\mathcal{Q}$. While this set is known to be convex \cite{tsirelson1993some}, no general explicit characterisation is currently available. Nonetheless, it can be efficiently approximated from the outside by a converging hierarchy of semidefinite programs, known as the Navascués-Pironio-Acín (NPA) hierarchy \cite{navascues2007bounding, navascues2008convergent}. In this work, we use the NPA hierarchy as a relaxation of the problem of characterising $\mathcal{Q}$. For simplicity, we will henceforth (somewhat abusively) identify $\mathcal{Q}$ with a fixed level of the NPA hierarchy. Note that using this relaxation to find the optimal PEF provides a lower bound on the solution of \eqref{eq:PEFoptnobeta}, since the new feasible set is a subset of the original one. This coincides with the intuition that giving the adversary a larger set of allowed behaviours still provides a valid, most likely suboptimal, entropy bound. 

Despite our efforts, this relaxation of optimisation \eqref{eq:PEFoptnobeta} is still challenging. It consists of optimising a logarithmic function over a convex domain of $F$. This domain is nonlinearly dependent on a variable $\mu$ which must belong to a chosen level of the NPA hierarchy (defined by semidefinite constraints). It is not known if one can directly incorporate the NPA constraints into \eqref{eq:PEFoptnobeta} in such a way that the resulting problem is solvable.
% The relaxation of \eqref{eq:PEFoptnobeta} thus becomes an optimisation of a logarithmic function over a convex domain defined by semidefinite constraints.Again, an efficient algorithm to tackle this problem is not available, and a further relaxation is required.
Therefore, we follow the approach in \cite{zhang2018certifying, knill2020generation} to further relax this problem.  We approximate the set $\mathcal{Q}$ from the outside by a convex polytope $\mathcal{P}_Q$. Then, due to convexity, it suffices to impose the PEF constraint on all the extreme points of this polytope, which we denote by the set $\textrm{Extr}(\mathcal{P}_Q)$. 
In the measurement independent scenario, this leads to the following relaxation of \eqref{eq:PEFoptnobeta}: 
\begin{equation}\label{eq:optPEFcond}	\begin{aligned}
			\max_F \quad & \sum_{c,z} p(c,z)\log F(c,z)\\
			\textrm{s.t.} \quad & F \geq 0\\
			&\sum_{c,z}F(c,z)\mu(c|z,e)p(z) \mu(d|z,e)^{\beta}\leq 1, \quad \forall \mu \in \textrm{Extr}(\mathcal{P_Q})\,,\\
		\end{aligned}
	\end{equation} 
where $p(z)=\mu(z)$ is assumed to be known or estimated by the user.

\section{Outer polytope approximations to the quantum set}
\label{sec:polytope}
	
We have seen in the last section that finding appropriate outer polytope approximations $\mathcal{P}_Q$ is a step toward obtaining effective entropy estimators \eqref{eq:entropy_estIID}. But what makes a good approximation? Ideally, polytopes $\mathcal{P}_Q$ should:
\begin{enumerate}
	\item Closely approximate the original quantum set $\mathcal{Q}$. This avoids the gross overestimation of the adversary's allowed behaviours, which can lead to unnecessarily suboptimal entropy estimators;
	\item Be simple enough such that \emph{(i)} it is possible to find all its extreme points and \emph{(ii)} the number of extreme points is not excessively large, allowing us to solve \eqref{eq:optPEFcond} in a reasonable amount of time. 
\end{enumerate}
	
Here we propose a method for designing polytopes $\mathcal{P}_Q$ that explicitly takes into account the typical conditional behaviour $p(C|Z)$, which allows these approximations to be effective and tractable. Our starting point is a polytope $\mathcal{P}_{\rm in}$ that contains $\qs$, for instance the set of non-signalling behaviours $\ns$ \cite{popescu1994quantum, BellReview2014}. We then obtain refined polytopes $\mathcal{P}_Q$ by iteratively eliminating supra-quantum behaviours\footnote{A behaviour is supra-quantum if it does not have a quantum realisation (see Section~\ref{sec:intro}).} that could significantly enhance the adversary's predictive power. For this, at each iteration, we tighten the approximation by imposing additional quantum Bell inequalities chosen to be tangent to $\qs$. The two main algorithms described in the following lead to PEFs that perform well in a variety of scenarios, as illustrated in the next section.
 	
 \begin{figure}[t]
 	\centering
 	\begin{subfigure}[t]{0.49\textwidth}
 		\centering
 		\includegraphics[width=0.94\linewidth]{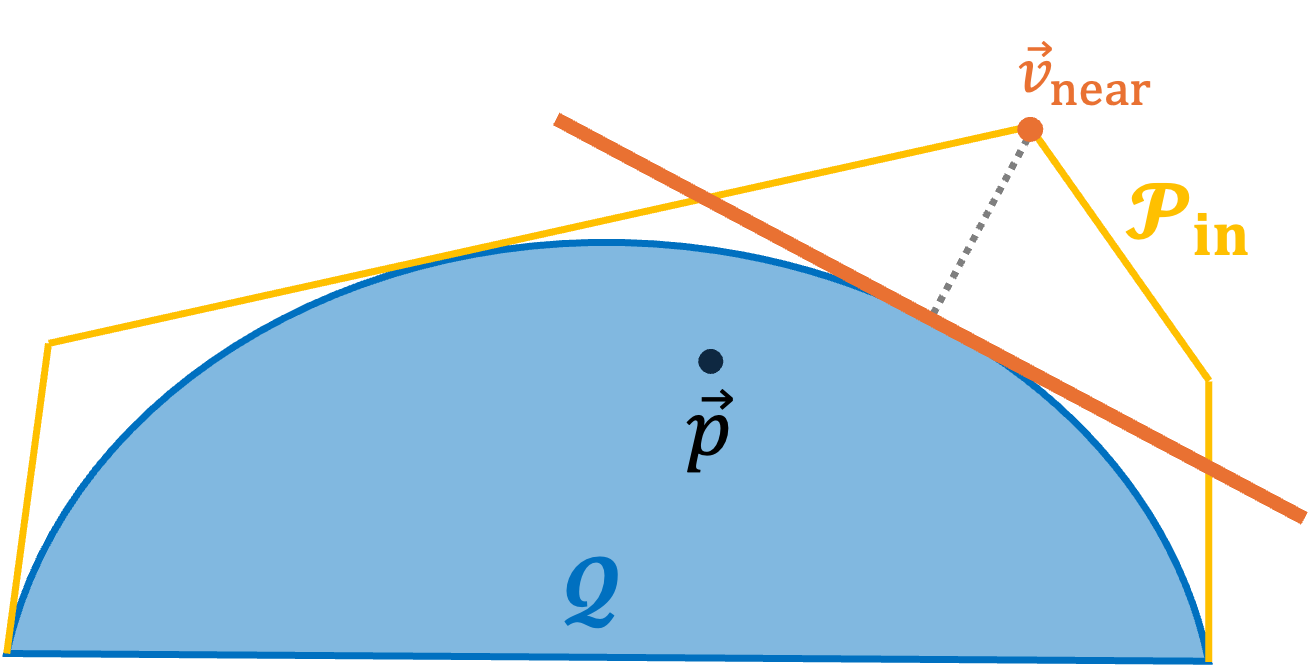}
 		\caption{Algorithm~\ref{alg:closestEP} using near non-quantum vertices}
 		\label{fig:diagram_closest_extreme}
 	\end{subfigure}
 	 \hfill
 	\begin{subfigure}[t]{0.49\textwidth}
 		\centering
 		\includegraphics[width=\linewidth]{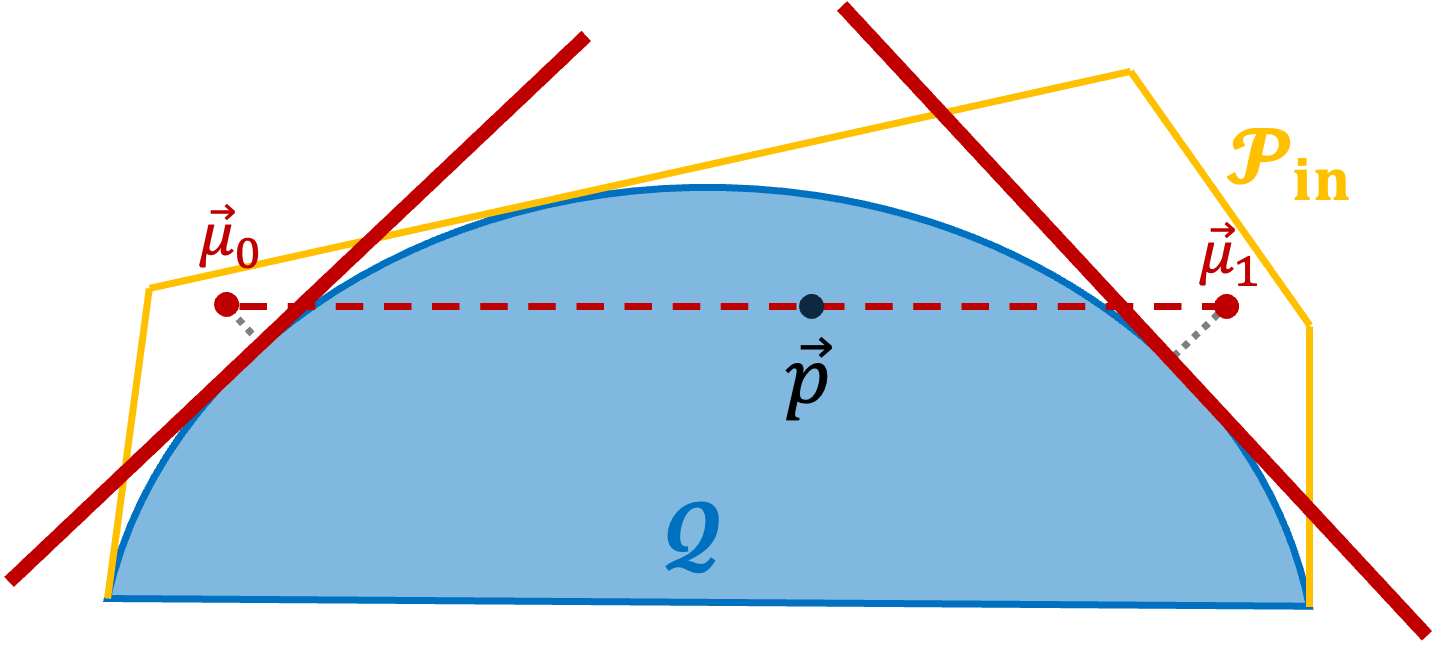}
 		\caption{Algorithm~\ref{alg:gp} using optimal guessing strategies}
 		\label{fig:diagram_min_entropy}
 	\end{subfigure}
 	\caption{\textbf{Illustrations of algorithms NearV (\ref{alg:closestEP}) and MaxGP ~(\ref{alg:gp}).} The blue region represents an approximation to the quantum set $\qs$ given by some level of the NPA hierarchy while the yellow lines represent the boundaries of a polytope $\poly_{\rm in}$ containing $\qs$.  Our aim is to generate a finer outer-polytope approximation $\poly_Q$ to $\qs$ taking into account the typical conditional behaviour $\vec{p}$ (black dot). 
 		(a) NearV (Algorithm~\ref{alg:closestEP}): Randomly choose one vertex $\vec{v}_{\rm near}$ from the set of $m$ nearest non-quantum vertices to $\vec{p}$ and generate a quantum Bell inequality $(\vec{b},\beta)$ (half-space delimited by the orange line, tangent to $\qs$) defined by the vector between $\vec{v}_{\rm near}$ and its closest $\vec{q}\in\qs$ (the grey dotted line). The output polytope $\poly_Q$ is the intersection of $\poly_{\rm in}$ with the quantum Bell inequality $(\vec{b},\beta)$.
 		(b) MaxGP (Algorithm~\ref{alg:gp}): Solve optimisation \eqref{eq:GP_A_givenxbar} for $\vec{p}$ and a set of allowed behaviours $\poly_{\rm in}$ and obtains the optimal adversarial guessing strategies $\{\vec{\mu}_\lambda\}_\lambda$ (in the illustration, red points with $|E|=2$). The algorithm uses each $\vec{\mu}_\lambda$ to generate a quantum Bell inequality (delimited by a red line) in a similar method to Algorithm~\ref{alg:closestEP}. The new polytope $\poly_Q$ is the intersection of $\poly_{\rm in}$ with at least one quantum Bell inequality. }
 	\label{fig:diagram_polytope_algorithms}
 \end{figure}

\subsection{Polytope approximations using nearest non-quantum vertices}
The first algorithm, NearV, aims at building an approximation $\mathcal{P}_Q$ by excluding the regions of $\mathcal{P}_{\rm in}$ nearest to selected supra-quantum vertices (see Figure~\ref{fig:diagram_polytope_algorithms}). To understand the intuition behind this approach, recall that while the user observes a behaviour $p(C|Z)$, an adversary having access to classical correlations may know any possible convex decomposition of this point into any behaviours of $\mathcal{P}_{\rm in}$. In particular, writing this polytope in its vertex (or V)-representation, $V=\{v_k(C|Z)\}_k$, we have
\begin{equation}\label{eq:cl_sideinfo}
	p(C|Z)=\sum_k \omega(k)v_k(C|Z)\,,
\end{equation}
where $\omega(k)\geq 0$ and $\sum_k\omega(k)=1$. 
On one hand, we would like to identify and eliminate the non-quantum vertices $v_k$ that appear in the decomposition \eqref{eq:cl_sideinfo}, i.e., that can be part of an adversary's strategy. These vertices are intuitively in the vicinity of $p(C|Z)$. On the other hand, removing supra-quantum extremal behaviours associated with smaller weights $\omega(k)$ (i.e., larger distances from $p(C|Z$)) avoids decompositions that contain a significant contribution of local vertices, which are always deterministic. 
Empirically, we find that an effective approach is to randomly select a vertex from the subset of the $m$ vertices closest to $p(C|Z)$ and then eliminate the supra-quantum region between this vertex and the quantum set $\mathcal{Q}$. This strategy has proven effective in reducing the adversary’s predictive power. The parameter $m$ is user-defined and depends on the setup.

We present the NearV algorithm using a vectorial representation of the space of behaviours, meaning that $\mu(C|Z)$ will  be represented by a vector $\vec{\mu}$ with components $\mu_{cz}$. A half-space $(\vec h ,\eta)$ is the set of vectors $\vec \mu$ such that
\begin{equation}\label{eq:halfspace}
	\vec h\cdot \vec \mu \leq \eta\,.
\end{equation}
A quantum Bell inequality $(\vec b,\beta)$ is a half-space which contains the quantum set $\qs$
\begin{equation}\label{eq:qBellineq}
	\vec b\cdot \vec \mu \leq \beta\,, \qquad \forall \vec \mu\in \qs \,.
\end{equation}

\begin{algorithm}[H]
	\caption{NearV: Generates a new polytope $\mathcal{P}_Q$ using the nearest non-quantum vertices to $p(C|Z)$.}\label{alg:closestEP}
	\begin{algorithmic}[1]
		\Require  $\vec{p}$ - typical conditional behaviour;  
		 $H= \{(\vec{h}_j,\eta_j)\}_j$ - $\poly_{\rm in}$ polytope in H-representation; 
		 $I$ - number of iterations; $m$ - size of list of nearest vertices
		\Ensure $\poly_Q$ in V-representation
	\For{$i \gets 1$ to $I$} 
	\State Express the polytope $H$ in its $V$-representation $V\gets \{\vec{v}_k\}_k$;
    \State Identify all the vertices $\{\vec{v}_{k'}\}_{k'}$ that are also in $\qs$ and define the list of quantum vertices $V_Q\gets \{\vec{v}_{k'}\}_{k'}$;
	\State Define the list of non-quantum vertices $V_{NQ}\gets V\backslash V_Q$;
    \State Compute the distance $d_{k}=||\vec v_{k}-\vec p||$ for every $\vec{v}_{k}\in V_{NQ}$ and define the list of the $m$ nearest vertices $D \gets \{(\vec{v}_l,d_l)\}_{l=1}^m$;
    \State  Select $\vec v_l \in D$ with probability $d_l^{-1}/(\sum_{k=1}^m d_k^{-1})$ and set $\vec{v}_{\rm near} \gets \vec{v}_l$;
	\State Find the behaviour $\vec q\in \qs$ that minimises $||\vec v_{\rm near}-\vec q||$;
	\State Define the Bell coefficients $\vec{b}\gets\vec v_{\rm near}-\vec{q}$;
	\State Find the quantum bound $\beta\gets\max_{\vec{\mu}\in\qs} \vec{b}\cdot\vec{\mu}$;
	\State Define the tighter polytope approximation $H\gets [H,(\vec{b},\beta)]$;
	\EndFor
	\State Express the polytope $H$ in its V-representation $V\gets \{\vec{v}_l\}_l$; 
	\State Output $V$
\end{algorithmic}
\end{algorithm}

Note that the NearV algorithm can easily be modified to take an input polytope $\poly_{\rm in}$ in its $V$-representation. A possible choice for $\poly_{\rm in}$ is the non-signalling set $\ns$, or any known polytope containing $\qs$. 
A method for finding the $H$-representation of $\ns$ can be found, for instance, in \cite{BellReview2014} and references therein. Vertex enumeration algorithms to find the V-representation of a convex polytope given its H-representation are widely available. We have arbitrarily chosen the total variation distance in our implementations, but this could be replaced by another efficiently computable distance. Likewise, one can use different  probability distributions to choose $\vec v_{\rm near}$ in Step 5. For instance, the uniform distribution also gave good results in some of our simulations. Finally, Steps 3, 7 and 9 require the use of the NPA relaxation algorithm.

\subsection{Polytope approximations using optimal guessing strategies}

Our second algorithm uses a more elaborate method to identify supra-quantum strategies that might significantly contribute to the predictive power of an adversary. The basic idea is to use an algorithm that finds the behaviours that maximise the adversary's probability of guessing $D$ for a typical distribution $p(C|Z)$. These optimal strategies are then used to define quantum Bell inequalities in a similar way to the NearV algorithm (see Figure~\ref{fig:diagram_polytope_algorithms}). Note that the adversarial strategies that maximise the guessing probability---equivalently minimise the min-entropy of $D$, see \eqref{def:Hmin} and \eqref{def:pguess}---are in general different from the adversarial strategies that minimise our single round entropy estimator from $(\log F)/\beta$ \eqref{eq:entropy_estIID} \footnote{We show in Appendix \ref{app:PEFmintradeoff} that the expected value of this estimator defines a min-trade-off function for the average conditional Shannon entropy.}. Nevertheless, as we will see later, removing the sets of supra-quantum behaviours identified by this approach significantly increases the amount of verifiable entropy using PEFs in various examples.

\paragraph{Algorithm to find optimal guessing strategies.} We start by reviewing the algorithm introduced by Nieto \textit{et al} \cite{nieto2014using}. It relies on the description of a general adversarial strategy to guess $D$, given that a particular input $z$ has been chosen. The adversary uses its side information $\Lambda$ to make a guess on $D$ and is successful whenever $\lambda=d$. This implies that it is enough to consider the number of strategies to be $|\Lambda|=|D|$ and the guessing probability can be written as
\begin{equation}
	P_{\rm guess}(D|z,\Lambda)=\max_{q_\Lambda,\{\mu_\lambda(C|Z)\}_\lambda }\left(\sum_d q_{\Lambda}(d)\mu_{d}(d|z)\right)\,,
\end{equation}
where $q_\Lambda$ is the probability distribution of $\Lambda$. Strategies $\mu_d(C|Z)=\mu(C|Z,\Lambda=d)$ belong to the set of allowed behaviours $\phyM$ and must be compatible with the behaviour $p(C|Z)$. 
We can absorb the normalisation factor $q_\Lambda(d)$ and define sub-normalised strategies $\tilde\mu_d$ as \begin{equation}
	\tilde\mu_d(C|Z)=q_{\Lambda}(d)\mu_d(C|Z)\,.
\end{equation}
The guessing probability can then be expressed as the following optimisation problem
\begin{align}
	\begin{split}	\label{eq:GP_A_givenxbar}
	P_{\rm guess}(D|z,\Lambda) = \max_{\{\tilde{\mu}_\lambda\}_\lambda}  & \sum_{d} \tilde{\mu}_d(d|z)\\
		\text{s.t. } & \sum_{d} \tilde{\mu}_d(C|Z) = p(C|Z)\,, \quad \forall \tilde{\mu}_d (C|Z)\in \widetilde{\phyM}\,,
	\end{split}
\end{align}
where $\widetilde{\phyM}$ is the set of sub-normalised allowed behaviours. 
The solution to this problem gives us also the set of optimal strategies $\{\mu_d(C|Z)\}_d$ for the chosen input $z$.

We are now in position to introduce our second algorithm to approximate the quantum set.

\begin{algorithm}[H]
	\caption{\textit{MaxGP}: Generates a polytope $\mathcal{P}_Q$ using the supra-quantum behaviours that maximise the guessing probability of $D$.}\label{alg:gp}
	\begin{algorithmic}[1]
		\Require  $\vec{p}$ - typical conditional behaviour; $p(Z)$ -  input distribution;
		$H= \{(\vec{h}_j,\eta_j)\}_j$ - $\poly_{\rm in}$ polytope in H-representation; 
		$I$ - number of iterations
		\Ensure $\poly_Q$ in V-representation
		\For{$i\gets1$ to $I$}
		\State Randomly pick a setting $z$ according to $p(Z)$;
		\State Obtain the behaviours $\{ \vec{\mu}_{\lambda}\}_{\lambda}$ that maximise $P_{\rm guess}(D|z,\Lambda)$ for the set of allowed behaviours $H$;
		\For{$\lambda\gets1$ to $|\Lambda|$}
		\If{$\vec{\mu}_\lambda\not\in\qs$}
		\State Find the behaviour $\vec q\in \qs$ that minimises $||\vec{\mu}_\lambda-\vec q||$;
		\State Define the Bell coefficients $\vec{b}\gets\vec{\mu}_\lambda-\vec{q}$;
		\State Find the quantum bound $\beta\gets\max_{\vec{\mu}\in\qs} \vec{b}\cdot\vec{\mu}$;
		\State Define the refined polytope approximation $H\gets [H,(\vec{b},\beta)]$;
		\EndIf
		\EndFor
		\EndFor
		\State Express the polytope $H$ in its V-representation $V\gets \{\vec{v}_l\}_l$; 
		\State  Output $V$
	\end{algorithmic}
\end{algorithm}

Note that we have chosen here to randomly pick the setting $z$ to compute the optimal strategies at every iteration of the algorithm. Other choices are possible, for instance picking $z$ that gives the highest $P_{\rm guess}(D|z,\Lambda)$, finding the strategies that maximise the guessing probability averaged over the inputs $z$ (see the approach in \cite{bancal2014more}), or simply choosing a fixed setting $z$.

\section{Randomness certification using new polytope approximations}
\label{sec:DIrandVerif}

In this section, we consider a variety of device-independent protocols for randomness certification and compute the amount of extractable entropy \eqref{eq:extractsmoothH} using the PE approach combined with polytope approximations $\mathcal{P}_Q$ generated by the algorithms introduced in the previous section. 

First, we consider the most standard protocols in which the entropy source $\mathbf{Z}$ is independent from the side information $E$, therefore relying on Bell tests. In the bipartite case, we analysed the global output entropy in noisy standard \cite{clauser1969proposed} and asymmetric \cite{acin2012randomness} CHSH correlations, and in  experimental data obtained from running loophole-free CHSH-Bell tests 
\cite{rosenfeld2017event, zhang2020experimental}. For the tripartite case, we studied the two-party output entropy in noisy Mermin correlations obtained through real experiments on a commercially available quantum computer (H1 from Quantinuum).  Our results are compared with the extractable entropy obtained in previous analyses with PE or with other commonly used methods, namely:
\begin{itemize}
	\item The PE approach using the original polytope approximations \cite{knill2020generation} and approximations introduced in \cite{bierhorst2020tsirelson};
	\item The randomness certification method for classical side information described in \cite{pironio2013security,nieto2014using,nieto2018severalBell}. This is one of the first approaches to successfully certify entropy in a non-IID setting;
	\item The EAT method, as introduced in \cite{Dupuis:2020EAT}, which considers quantum side information. This result (and its variations) is the most commonly used method to certify entropy under coherent attacks. 
\end{itemize}

We also consider protocols where the input entropy source might be correlated  with the side information $E$, meaning that inputs may be correlated to the state of the device in a manner known to the adversary. These are commonly referred to as \emph{randomness amplification} protocols. We analyse the correlations arising in a practical protocol proposed in  \cite{ramanathan2018practical} and compare with their results. 

Note that when analysing theoretical noisy correlations or experimental data with no prior calibration data, we consider the typical behaviour $p$ to coincide with the observed $p_{\rm obs}$.

\subsection{Protocols using an independent entropy source}

When the inputs $\mathbf{Z}$ are chosen independently of the state of the device, finding a PEF adapted to the typical behaviour $p(C,Z)$ involves solving the optimisation problem \eqref{eq:optPEFcond} for a set of allowed conditional behaviours, as described in Section~\ref{sec:findPEF}. In this section, inputs are therefore independent of the side information $p(Z):=\mu(Z|e)=\mu(Z)$. We have chosen uniformly distributed inputs $p(Z)=1/|Z|$ for simplicity, but other choices would be possible (for instance, we could have chosen distributions biased towards a particular input, like in spot-checking protocols).

\subsubsection{Standard and asymmetric noisy CHSH correlations}\label{sec:noisyCHSH}

Consider the bipartite Bell scenario, with $C=AB$ and $Z=XY$, and the asymmetric CHSH inequality \cite{clauser1969proposed,acin2012randomness}
\begin{equation}
	\label{eq:tilted_ineq}
	\textrm{CHSH}_\alpha:=\alpha\Var{\mathsf{A}_0\mathsf{B}_0} + \alpha\Var{\mathsf{A}_0\mathsf{B}_1} + \Var{\mathsf{A}_1\mathsf{B}_0} - \Var{\mathsf{A}_1\mathsf{B}_1} \leq 2\alpha\,,
\end{equation}
where the local observables $\mathsf{A}_x$ and $\mathsf{B}_y$ can take values $\pm 1$, $\Var{\cdot}$ represents their expected value, for instance $\Var{\mathsf{A}_x}=\sum_a ap(a|x)$, and $2\alpha$ is the classical bound of the expression. For $\alpha=1$, we recover the standard CHSH Bell inequality while $\alpha>1$ corresponds to the asymmetric case. The maximum quantum Bell violation of \eqref{eq:tilted_ineq} is achieved when
\begin{equation}\label{eq:idealcorr_asym}
	\begin{split}
		\Var{\mathsf{A}_x}&=\Var{\mathsf{B}_y}=0\,,\\\Var{\mathsf{A}_0\mathsf{B}_y}=\frac{\alpha}{\sqrt{1+\alpha^2}}\,,&\qquad \Var{\mathsf{A}_1\mathsf{B}_y}=\frac{(-1)^y}{\sqrt{1+\alpha^2}}\,,
	\end{split}
\end{equation}
which corresponds to the conditional behaviour $p^{*}_{\alpha}(A,B|X,Y)$ via $p(a,b|x,y)=1/4(1+\Var{\mathsf{A}_x}+\Var{\mathsf{B}_y}+ab \Var{\mathsf{A}_x\mathsf{B}_y})$. 

Here we analyse protocols with typical conditional behaviours that are convex mixtures of the optimal behaviour \eqref{eq:idealcorr_asym} and white noise
\begin{equation}\label{eq:noisyCHSH}
	p_{\alpha,w}(A,B|X,Y)=(1-w) p^{*}_{\alpha}(A,B|X,Y)+ w\frac{1}{4}\,,
\end{equation}
for $\alpha=1$ (standard CHSH) and $\alpha=8$, and a variety of noise levels $w\in[0,1]$.

\paragraph{Obtained results.} Figures \ref{fig:bi_randVeri_005WN}-\ref{fig:bi_tilted_07WN} depict lower bounds on the extractable entropy per round $\frac{1}{n} H_{\min}^{\varepsilon/p_{Acc}}(\mathbf{AB}|\mathbf{XY},E)$ \eqref{eq:extEntIID} as a function of the number of rounds $n$ of the protocol. We considered acceptance thresholds \eqref{eq:accept_threshold} with $\delta=0$ and security parameter $\varepsilon=2^{-128}$. Note that we used $\poly_{\rm in}=\ns \cap \textrm{CHSH} \leq 2\sqrt 2$, i.e. the intersection between $\ns$ and the behaviours that do not violate CHSH beyond its quantum bound, as the input for the NearV and MaxGP algorithms, which corresponds to the set of allowed behaviours used in \cite{knill2020generation}. % For the noisy asymmetric CHSH correlations (\ref{fig:bi_tilted}), we also analyse our approach by 

\begin{figure}[t]
	\centering
	\begin{subfigure}[t]{0.327\linewidth}
		\centering
		\includegraphics[width=\linewidth]{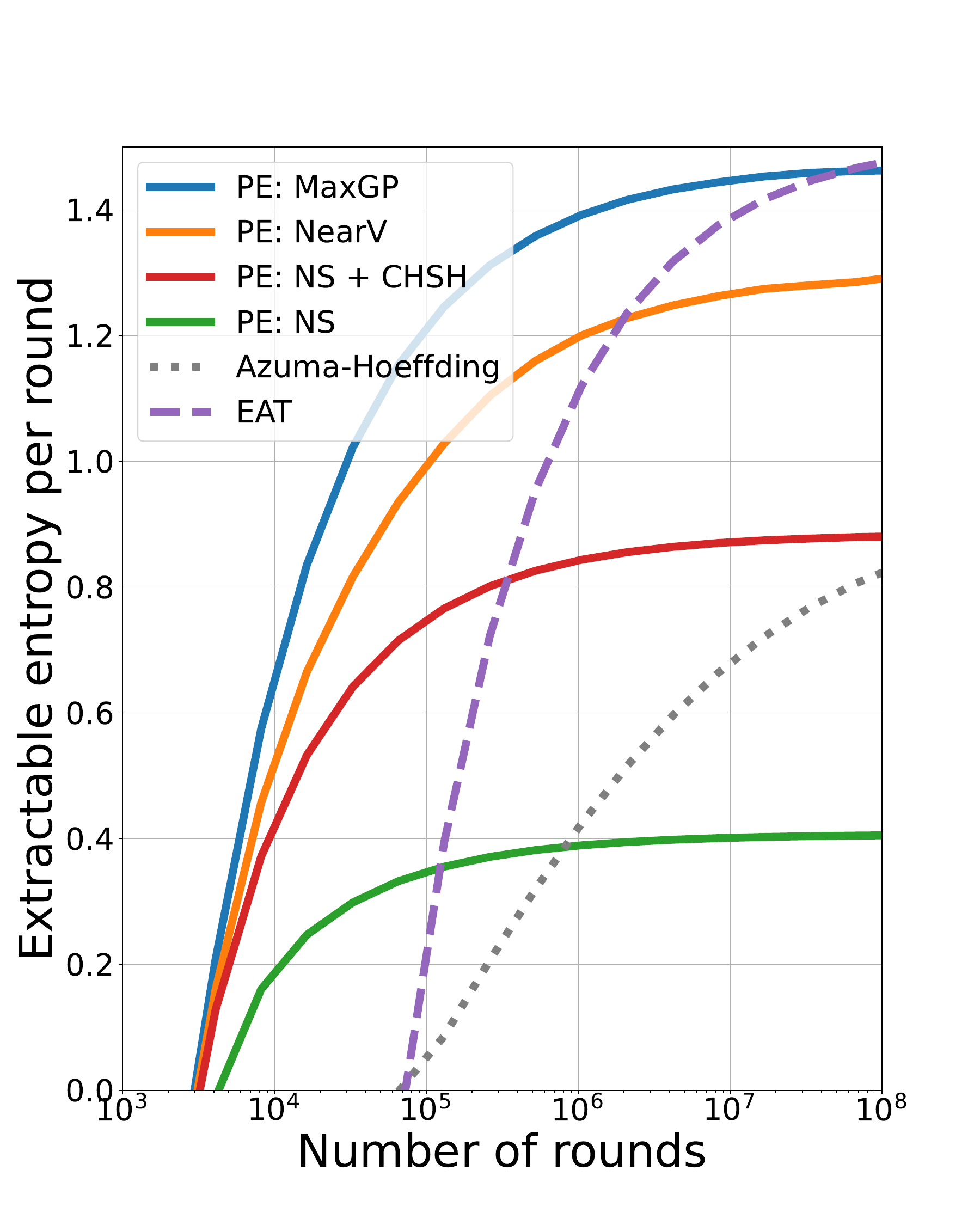}
		\caption{w=0.5\% }
		\label{fig:bi_randVeri_005WN}
	\end{subfigure}
	\hfill
	\begin{subfigure}[t]{0.327\linewidth}
		\centering
		\includegraphics[width=\linewidth]{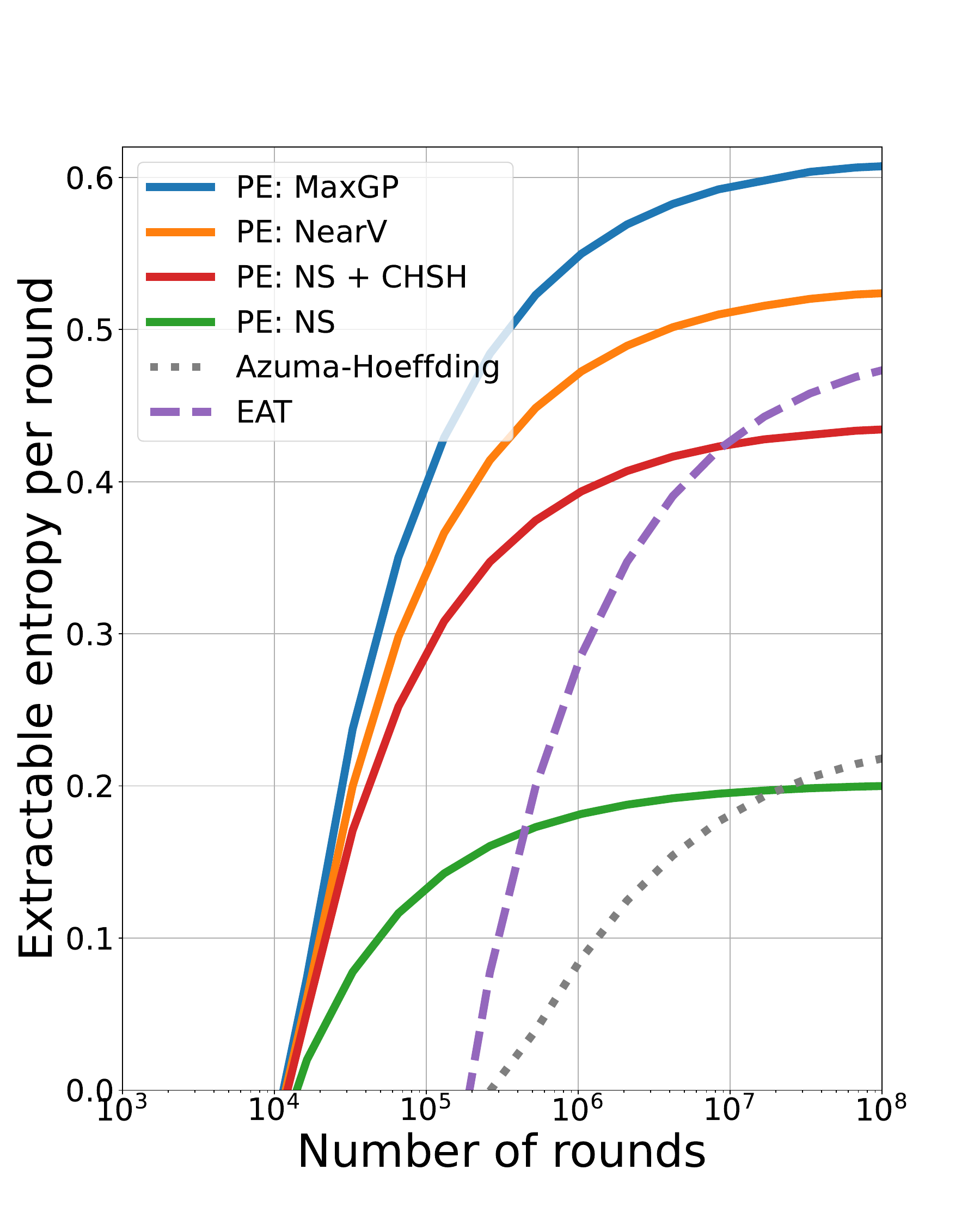}
		\caption{w=15\%}
		\label{fig:bi_randVeri_15WN}
	\end{subfigure}
	\hfill
	\begin{subfigure}[t]{0.327\linewidth}
		\centering
		\includegraphics[width=\linewidth]{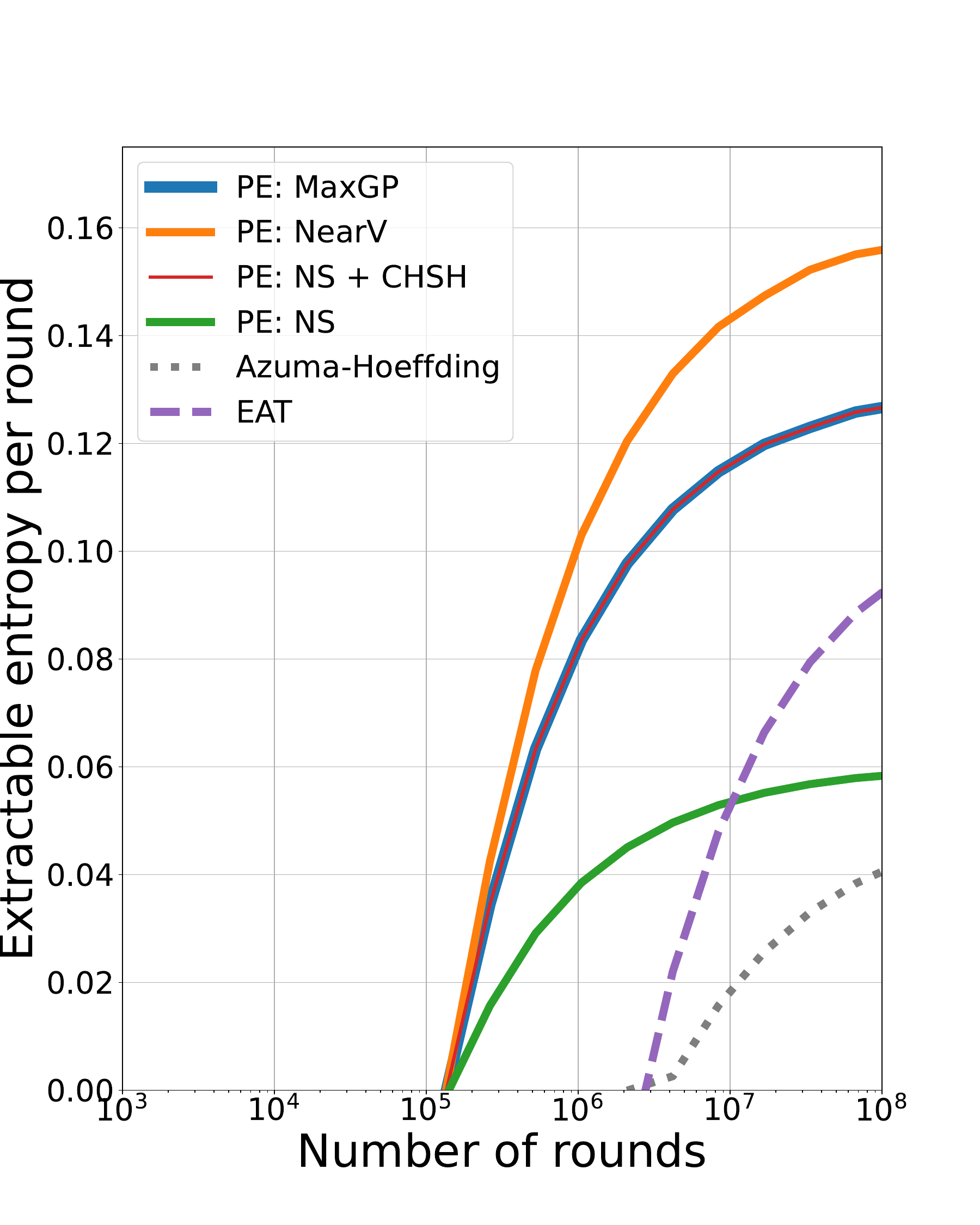}
		\caption{w=25\%}
		\label{fig:bi_randVeri_25WN}
	\end{subfigure}
	\caption{\textbf{Extractable entropy rates from noisy CHSH correlations for different levels of white noise.} Solid lines are obtained through the PE technique using different sets of allowed conditional behaviours: (orange) $\mathcal{P}_Q$ from 10 iterations of NearV (Algorithm~\ref{alg:closestEP} with $m=10$); (blue) $\mathcal{P}_Q$ from 10 iterations of MaxGP (Algorithm~\ref{alg:gp}); (green) $\poly_Q=\ns$; and (red) $\poly_Q=\ns \cap \textrm{CHSH} \leq 2\sqrt 2$, which coincides with set used in \cite{knill2020generation}. The purple dashed line is obtained using the EAT approach \cite{arnon2018practical} and the grey dotted line represents the method \cite{pironio2013security,nieto2014using} using the Azuma-Hoeffding inequality. The figure is produced using level 2 of the NPA hierarchy.}
	\label{fig:bi_randVeri}
\end{figure}

\begin{figure*}[t]
	\centering
	\begin{subfigure}[t]{0.327\linewidth}
		\centering
		\includegraphics[width=\linewidth]{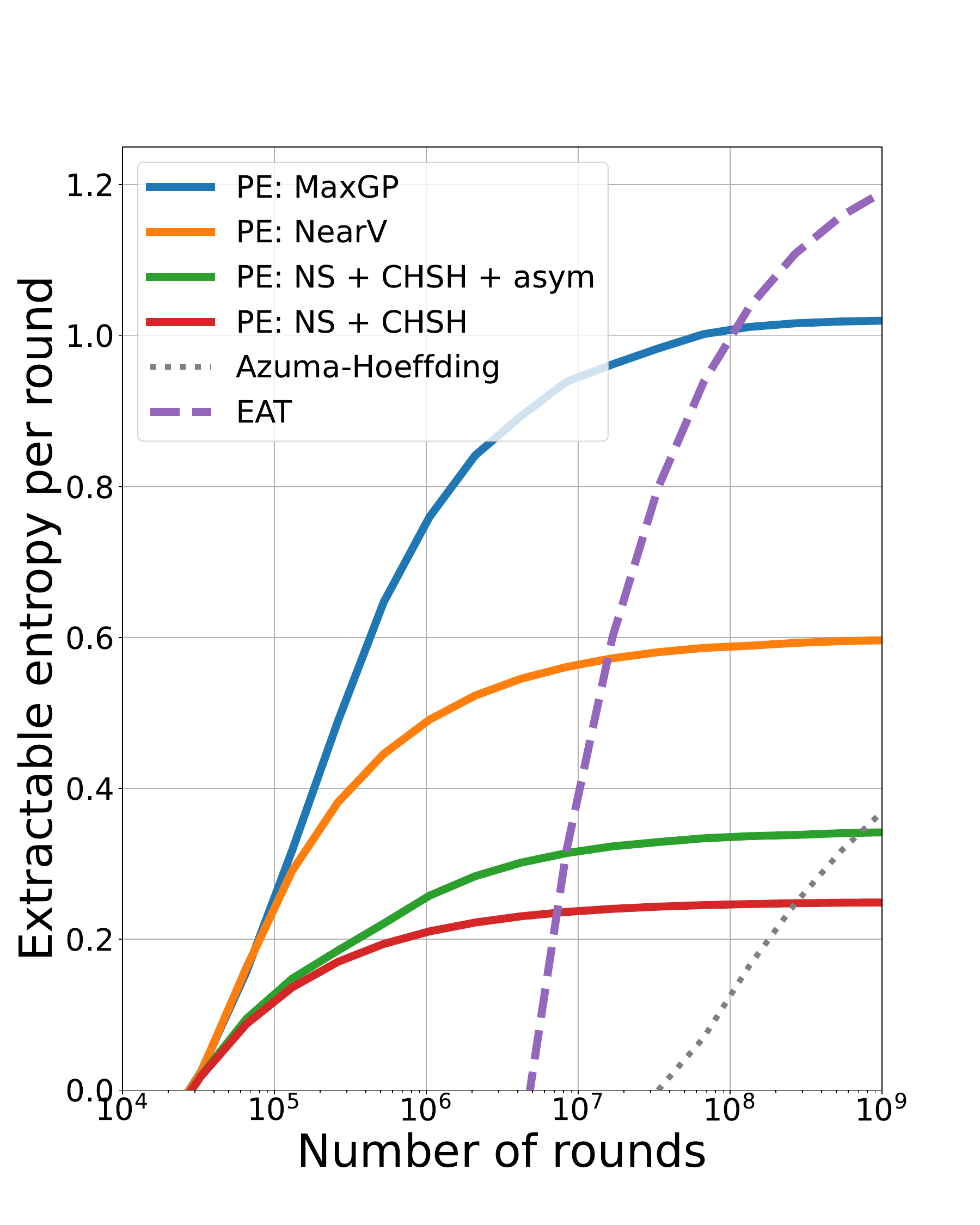}
		\caption{w=0.1\%}
		\label{fig:bi_tilted_001WN}
	\end{subfigure}
	\hfill
	\begin{subfigure}[t]{0.327\linewidth}
		\centering
		\includegraphics[width=\linewidth]{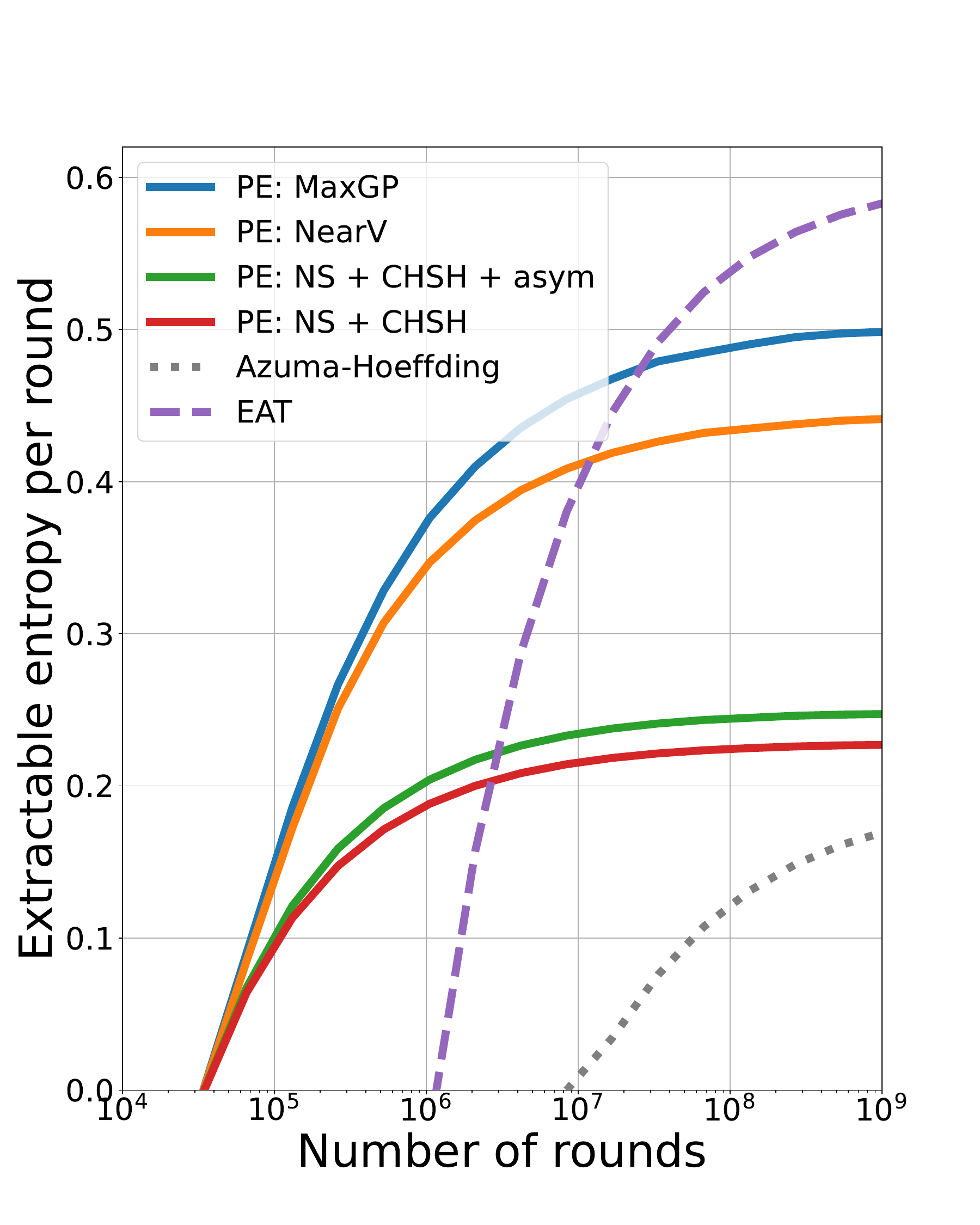}
		\caption{w=1\%}
		\label{fig:bi_tilted_01WN}
	\end{subfigure}
	\hfill
	\begin{subfigure}[t]{0.327\linewidth}
		\centering
		\includegraphics[width=\linewidth]{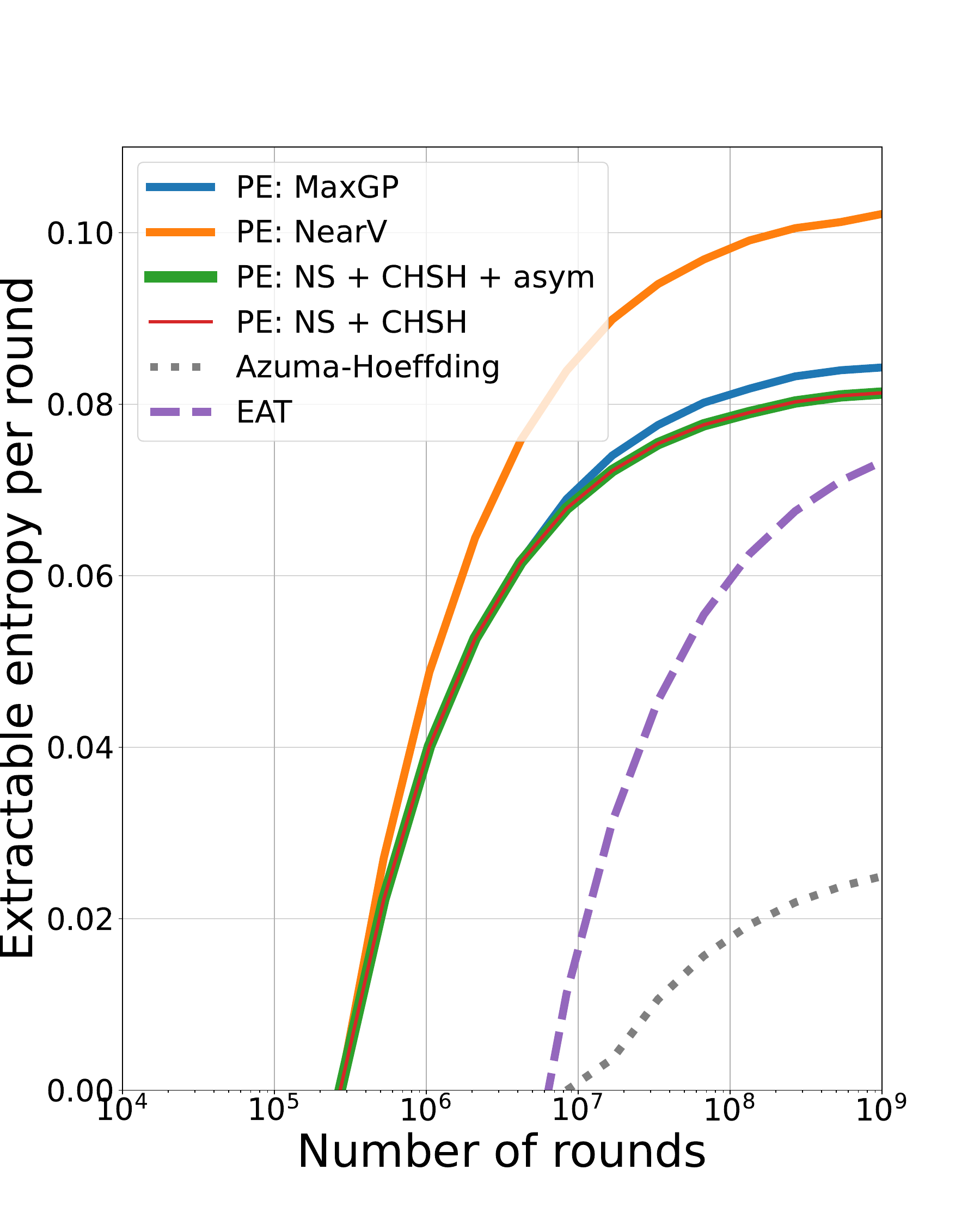}
		\caption{w=7\%}
		\label{fig:bi_tilted_07WN}  
	\end{subfigure}
	\caption{\textbf{Extractable entropy rates from noisy asymmetric CHSH correlations ($\alpha=8$) for different levels of white noise.} Solid lines are obtained through the PE technique using sets of allowed conditional behaviours: (orange) $\mathcal{P}_Q$ from 10 iterations of NearV (Algorithm~\ref{alg:closestEP} with $m=10$); (blue) $\mathcal{P}_Q$ from 10 iterations of MaxPG (Algorithm~\ref{alg:gp}); (red) - $\poly_Q=\ns \cap \textrm{CHSH} \leq 2\sqrt 2$ (set defined according to the approach in \cite{knill2020generation}); and (green) - $\poly_Q=\ns \cap \textrm{CHSH} \leq 2\sqrt 2 \cap \textrm{CHSH}_{\alpha=8}\leq TB$ (Tsirelson's bound) (set used in \cite{bierhorst2020tsirelson}). The purple dashed line is obtained using the EAT approach \cite{arnon2018practical}, and the grey dotted line uses the method in \cite{pironio2013security,nieto2014using}, which uses the Azuma-Hoeffding inequality. The figure is produced using level 2 of the NPA hierarchy.}
	\label{fig:bi_tilted}
\end{figure*}

\subsubsection{Noisy correlations from loophole-free CHSH-Bell tests}

We report here in Table~\ref{tab:Mermin_data} the improvements on the extractable entropy rates for data from loophole-free CHSH-Bell tests reported in \cite{rosenfeld2017event, zhang2020experimental}. In particular,
note that we obtained almost twice the entropy from the data in \cite{zhang2020experimental}, while previous attempts with refined polytopes \cite{bierhorst2020tsirelson} only achieved a marginal improvement.  
%%%%%%%%%%%%%%%%%%%%%%%%%%%
\begin{table}[H]
    \centering
    \begin{tabular}{|c|c||c|c|c|c|c|}
    \hline
    \multicolumn{2}{|c||}{} & & & & Extractable & Extractable\\
    \multicolumn{2}{|c||}{Data set} & $n$ & $\textrm{CHSH}_{\alpha=1}$ & $\varepsilon$ & entropy rate & entropy rate\\
    \multicolumn{2}{|c||}{} &&&& before$^*$ & with our methods\\
    \hline
    \hline
    \multicolumn{2}{|c||}{\cite[Table I]{knill2020generation}, \cite{rosenfeld2017event}} & 27,683 & 2.1756226 & $2^{-32}$ & 0.04119278 & 0.05010119\\
    \hline
    \multirow{5}{*}{\cite{zhang2020experimental}} & Instance 1 & 23,194,880 & 2.0011386 & $2^{-64}$ &0.00033189 & 0.00056707\\
    & Instance 2 & 37,649,791 & 2.0011957 & $2^{-64}$ & 0.00021068 & 0.00047694\\
    & Instance 3 & 28,537,805 & 2.0010692 & $2^{-64}$ & 0.00030063 & 0.00052014\\
    & Instance 4 & 28,319,978 & 2.0010454 & $2^{-64}$ & 0.00030474 & 0.00050788\\
    & Instance 5 & 27,219,887 & 2.0011109 & $2^{-64}$ & 0.00028706 & 0.00050131\\
    \hline
    \end{tabular}
    \caption{Extractable entropy rates from loophole-free CHSH-Bell tests \cite{rosenfeld2017event,zhang2020experimental}. Our method improves the extractable entropy rates compared to those previously reported. 10 iterations of NearV (Algorithm~\ref{alg:closestEP} with $m=10$) is used to obtain the result for \cite[Table I]{knill2020generation}\cite{rosenfeld2017event}, and MaxGP (Algorithm~\ref{alg:gp}) is used to obtain the results for \cite{zhang2020experimental}. 
    $^*$For \cite{zhang2020experimental} these entropy rates are calculated with the optimal PEFs reported in their Table IV, Appendix IV, and for \cite[Table I]{knill2020generation} we obtain the rate with $\poly_Q=\ns \cap \textrm{CHSH} \leq 2\sqrt 2$ which was used in their analysis.
    }
    \label{tab:loophole_free_data}
\end{table}
%%%%%%%%%%%%%%%%%%%%%%%%%%%

\subsubsection{Noisy Mermin correlations}

A well-studied Bell inequality in the tripartite scenario is the Mermin inequality 
\begin{equation}
		\label{eq:mermin_ineq}
		\textrm{M}=\Var{\mathsf{A}_0\mathsf{B}_0\mathsf{C}_0} - \Var{\mathsf{A}_0\mathsf{B}_1\mathsf{C}_1} - \Var{\mathsf{A}_1\mathsf{B}_0\mathsf{C}_1} - \Var{\mathsf{A}_1\mathsf{B}_1\mathsf{C}_0} \leq 2\,
\end{equation}
where local observables $\mathsf{A}_x$, $\mathsf{B}_y$  and $\mathsf{C}_z$ can take values $\pm 1$. The maximum quantum violation (and algebraic bound) of this inequality is $\textrm{M}=4$, which can be achieved by measuring the GHZ state $\frac{1}{\sqrt 2}(\ket{000}+\ket{111})$ with the Pauli operators $\mathsf{A}_0=\mathsf{B}_0=\mathsf{C}_0=\sigma_x$ and $\mathsf{A}_1=\mathsf{B}_1=\mathsf{C}_1=\sigma_y$. 

To showcase our method and the effectiveness of our algorithms, we have implemented the Mermin-Bell test on an ion-trap quantum computer (H1 from Quantinuum \cite{QuantinuumH1})\footnote{While one cannot guarantee the non-signalling condition in experiments performed on a single quantum computer, the negligible levels of cross-talk in H1 support the assumption that local measurements are indeed independent across different locations.} with $n=4\times 10^4$ rounds, from which we obtained a behaviour achieving $\textrm{M}=3.928$ shown in Table~\ref{tab:Mermin_data}. We then analyse the extractable entropy of outputs $\mathbf{D=AB}$ using the typical behaviour estimated by \eqref{eq:freq} and regularised according to \cite{lin2018device}. The results depicted in Figure~\ref{fig:NetEntR_vs_N_tri} include the extractable entropy from the same behaviour as if it were obtained with a different $n$.

\begin{table}[h!]
    \centering
    \begin{tabular}{|c||c|c|c|c|c|c|c|c|}
        \hline
        \backslashbox{abc}{xyz}
        &000&001&010&011&100&101&110&111 \\\hline\hline
        000 & 0.0044 & 0.1258 & 0.1302 & 0.2514 & 0.1282 & 0.2414 & 0.2484 & 0.1236 \\\hline
        001 & 0.247 & 0.129 & 0.127 & 0.0034 & 0.1254 & 0.0024 & 0.0024 & 0.1358 \\\hline
        010 & 0.244 & 0.1252 & 0.115 & 0.0014 & 0.124 & 0.0018 & 0.0016 & 0.1242 \\\hline
        011 & 0.002 & 0.126 & 0.125 & 0.248 & 0.1194 & 0.2546 & 0.2468 & 0.1292 \\\hline
        100 & 0.2536 & 0.1204 & 0.1252 & 0.002 & 0.1264 & 0.003 & 0.0024 & 0.125 \\\hline
        101 & 0.0014 & 0.1256 & 0.1232 & 0.2352 & 0.1202 & 0.2472 & 0.2518 & 0.119 \\\hline
        110 & 0.0018 & 0.1228 & 0.1304 & 0.2556 & 0.1318 & 0.2486 & 0.2446 & 0.1208 \\\hline
        111 & 0.2458 & 0.1252 & 0.124 & 0.003 & 0.1246 & 0.001 & 0.002 & 0.1224 \\\hline
    \end{tabular}
    \caption{
    Estimated probabilities $p(A,B,C|X,Y,Z)$ from running the Mermin game in H1 from Quantinuum, with $n=4\times10^4$. Raw data before regularisation, achieving $\textrm{M}=3.928$.
    }
    \label{tab:Mermin_data}
\end{table}

\begin{figure}[t]
	\centering
	\includegraphics[width=0.6\textwidth]{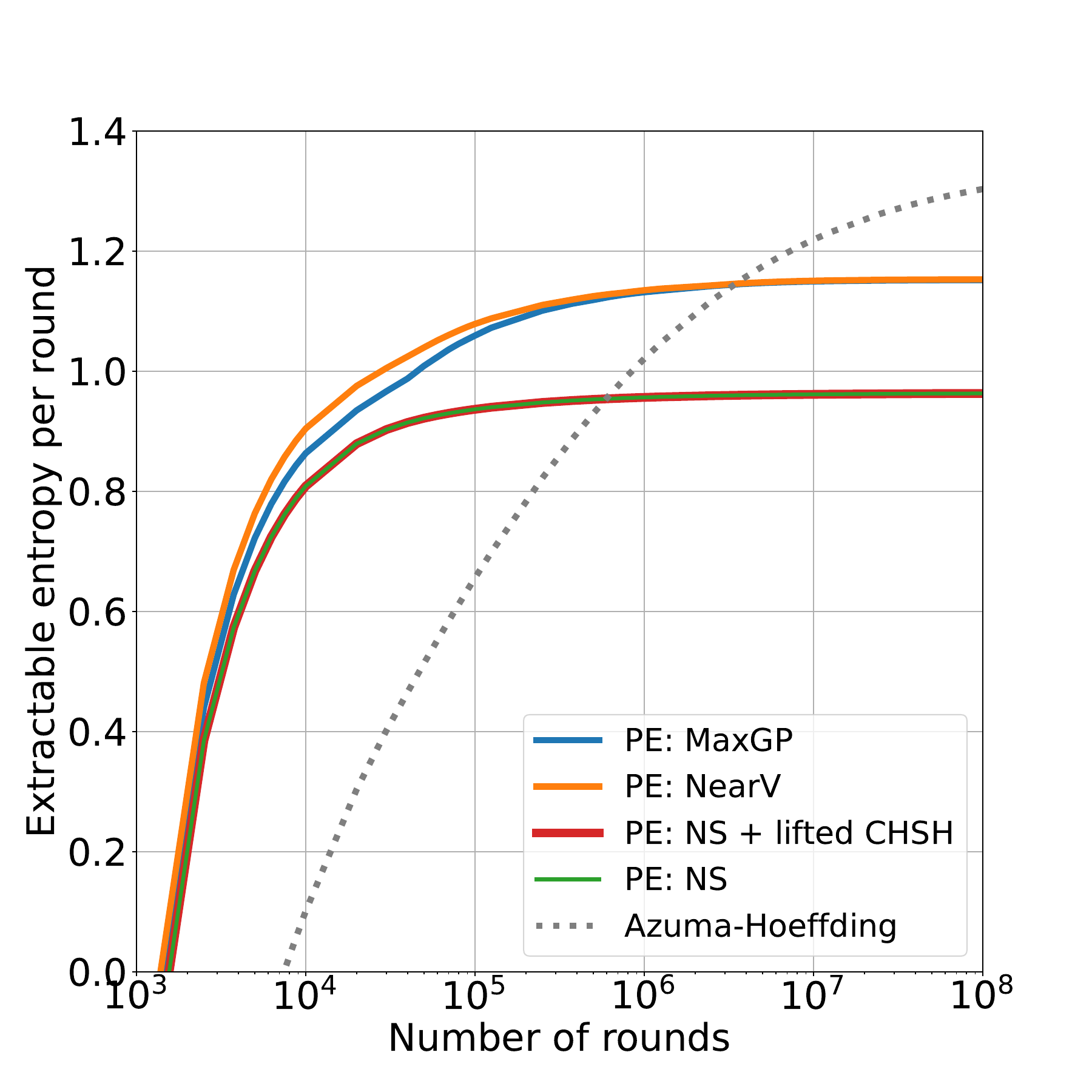}
	\caption{\textbf{Extractable entropy rates from noisy Mermin correlations obtained from Table \ref{tab:Mermin_data} as a function of the number of rounds $n$.} Solid lines correspond to lower bounds computed with the PE technique using  $\poly_Q$ obtained through: (orange) NearV algorithm (Algorithm~\ref{alg:closestEP} with $m=1$); (blue) MaxGP algorithm (Algorithm~\ref{alg:gp});  (green) $\poly_Q=\ns$; (red) $\poly_Q=\ns \cap  (\textrm{lifted-CHSH}\leq 2\sqrt{2}$), where the lifted-CHSH is a generalisation of CHSH for 3 parties \cite{pironio2005lifting}. The grey dotted line is obtained by approach \cite{pironio2013security} using analytical optimal $P_{\rm guess}$ as a function of the Mermin value \cite{Woodhead2018randomnessversus}. The figure is produced using level 2 of the NPA hierarchy.}
	\label{fig:NetEntR_vs_N_tri}
\end{figure}

\subsection{Protocols using an entropy source correlated with the device} \label{sec:DIrandAmp}
Randomness amplification protocols allow one to obtain certified output entropy even if inputs are chosen from a weakly random source that is partly correlated with the quantum device. However, this comes at the price of more complex security analysis and lower output entropy rates. While one would expect sacrificing efficiency to achieve higher levels of security, how much entropy remains uncertified due to a suboptimal security analysis?

In this section, we show how to use the PE framework to lower-bound the extractable entropy of randomness amplification protocols that use a Santa Vazirani (SV)-source~\cite{santha1986generating} as an input. These entropy sources produce bits $S_i$ that, at every round, deviate from being perfectly random up to some bias parameter $0\leq \delta < 1/2$,
\begin{equation}\label{def:SVsource}
	\frac{1}{2}-\delta\leq \mu(S_i|S_{i-1},\ldots, S_{1},e)\leq \frac{1}{2}+\delta\,, \quad \forall i\,.
\end{equation} 
We illustrate our method by analysing the correlations arising in the bipartite protocol in \cite{ramanathan2018practical} based on the Hardy paradox \cite{HardyParadox}.

\subsubsection{Algorithm to find PEFs for correlated input source}

In the general case where the input randomness can be correlated with the quantum device, $\mu(Z|e)$ is no longer known to the user and the simplified optimisation \eqref{eq:optPEFcond} to find a PEF no longer holds. We must then go back to \eqref{eq:PEFoptnobeta}, which we reproduce here:

	\begin{align} 
		\max_F\quad & \sum_{c,z} p(c,z)\log F(c,z) \\
		\textrm{s.t.} \quad & F \geq 0\\
		&\sum_{c,z}\mu(c,z|e)F(c,z)\mu(d|z,e)^\beta\leq 1\,,\quad \forall \mu(C,Z|e) \in \phyM\,.\label{eq:const2PEF}
	\end{align}
As before, solving this optimisation requires the set $\phyM$ to be a polytope, in which case it is sufficient to impose the constraint \eqref{eq:const2PEF} only on the extreme points of this set. Given that we consider quantum adversaries, $\phyM$ is in general not a polytope, and we must find an outer-polytope approximation $\mathcal{P}_{\phyM}$ of $\phyM$. Although the allowed behaviours are now joint input-output distributions, this task will use the methods introduced in Section~\ref{sec:polytope} as a sub-routine.

We consider here a bipartite protocol where, at each round, the inputs of $\mathsf{A}$ and $\mathsf{B}$ are defined by two samples of a SV-source \eqref{def:SVsource}\footnote{We still require that the input of each party is not leaked to the other party; otherwise any distribution can be reproduced by classical strategies.}. Therefore, the distribution of Alice's input $X$ at any round belongs to the polytope $\mathcal{SV}_{\delta}$  defined by 
\begin{equation} \label{eq:1SVpolytope}
	\frac{1}{2}-\delta\leq \mu(X|\lambda)\leq \frac{1}{2}+\delta\,,
\end{equation}
where $\lambda$ represents all the side information available to the adversary in that round, i.e., the information of past samples of the SV-source and any additional information $e$ available. In full generality, we assume that Bob's input $Y$ is generated after $X$. The probability distribution of $Y$ also belongs to the polytope $\mathcal{SV}_{\delta}$ given by 
\begin{equation} \label{eq:1SVpolytopeB}
	\frac{1}{2}-\delta\leq \mu(Y|\lambda')\leq \frac{1}{2}+\delta\,,
\end{equation}
where the information $\lambda'$ contains both $\lambda$ and $x$.
The following theorem allows us to combine the sets of allowed input distributions for each party into a single convex polytope\footnote{We could have also considered an SV-source where at each round the condition  $l\leq \mu(X,Y|\lambda)\leq h$ holds. The vertices of this polytope are described in Lemma 1 of \cite{putz2016measurement}.}.
\begin{theorem}
	Consider two sets of inputs $X$ and $Y$ modelled by two samples of a $\delta$-SV-source \eqref{def:SVsource}, i.e., each following probability distributions  belonging to a convex polytope $\mathcal{SV}_{\delta}$ according to \eqref{eq:1SVpolytope} and \eqref{eq:1SVpolytopeB}, respectively.  Assume, without loss of generality, that  $X$ is sampled before $Y$. Then the set of allowed joint input distributions $\mathcal{SV}^2_{\delta}$ for both inputs is a polytope with vertices $u_{kk'}(X,Y)=v_k(X)v'_{k'}(Y)$, for $k,k'=0,1$, where $v_k(X)=\frac{1}{2}\left(1+(-1)^{k+X}\delta\right)$ are vertices of $\mathcal{SV}_{\delta}$ for party $\mathsf{A}$, and analogously for party $\mathsf{B}$. 
\end{theorem}
\begin{proof}
  By construction, the set of joint input distributions is
	\begin{equation}\label{eq:2SVpoly1}
		\mathcal{SV}^2_{\delta}=\{\mu(X,Y|\lambda):\mu(X,Y|\lambda)=\mu(X|\lambda)\mu(Y|X,\lambda), \mu(X|\lambda)\in \mathcal{SV}_{\delta},\mu(Y|x,\lambda)\in \mathcal{SV}_{\delta},\forall x\}\,,
	\end{equation}
    where $\mu(Y|x,\lambda)\in \mathcal{SV}_{\delta}$ following $\eqref{eq:1SVpolytopeB}$.
    
    Recall that a point $\mu$ belongs to a polytope $\mathcal{P}$ if and only if it can be written as a convex combination of its vertices $v_i$, 
	\begin{equation}
		\mu=\sum_i\omega(i)v_i\,,
	\end{equation}
	where $\omega\geq 0$ and $\sum_i\omega(i)=1$.
	Then we can express \eqref{eq:2SVpoly1} as 
	\begin{equation}
			\begin{split}
				\mathcal{SV}^2_{\delta}=\{\mu(X,Y|\lambda):&\mu(X,Y|\lambda)=\sum_{k,k'}\omega(k)\omega'(k')v_k(X)v'_{k'}(Y)\,, \\  &\omega,\omega'\geq 0\,, \sum_k\omega(k)=\sum_{k'}\omega'(k')=1\}\,.
		\end{split}
	\end{equation}
Since $\omega(k)\omega'(k')\geq 0$ and $\sum_{k,k'}\omega(k)\omega(k')=1$, we have that $\mathcal{SV}^2_{\delta}$ is a convex polytope with vertices $v_k(X)v'_{k'}(Y)$.
\end{proof}

As before, we approximate the set of conditional behaviours $\qs$ by a polytope $\mathcal{P}_Q$. This implies that the set $\mathcal{P}_\phyM$ of allowed behaviours is also a polytope with vertices $u_{ij}(CZ)=v_i(C|Z)v'_j(Z)$, where $\{v_i(C|Z)\}_i=\textrm{Extr}(\poly_Q)$ and $\{v'_j(Z)\}_j=\textrm{Extr}(\mathcal{SV}^2_{\delta})$. 
This result is analogous to the previous theorem and its proof is based on the decomposition $\mu(CZ|e)=\mu(C|Ze)\mu(Z|e)$.

We can then replace the constraint \eqref{eq:const2PEF} by the condition\footnote{We omitted conditioning on $e$ for simplicity.}
\begin{equation}\label{eq:const2PEFnew}
	\sum_{cz}\mu(c|z)\mu(z)F(c,z)\mu(d|z)^\beta\leq 1\,,\quad \forall\left(\mu(C|Z) \in\textrm{Extr}(\mathcal{P}_Q)\wedge  \mu(Z) \in\textrm{Extr}(\mathcal{SV}^2_{\delta})\right)\,,
\end{equation}
which is analogous to the last constraint in \eqref{eq:optPEFcond}, except that now it unfolds as a condition for each vertex of $\poly_Q$ and $\mathcal{SV}^2_{\delta}$.

\subsubsection{Noisy Hardy correlations}
In the scenario where inputs are correlated to the device, Bell inequalities are replaced by measurement dependent locality (MDL) inequalities, introduced in \cite{putz2016measurement}, as tests of non-locality. A protocol for randomness amplification will then require correlations that violate some MDL inequality. In~\cite{ramanathan2018practical}, the authors use a bipartite Hardy paradox \cite{HardyParadox} together with SV-sources \eqref{eq:1SVpolytope} to select the inputs for each party, and build the following MDL inequality on distributions $p(A,B,X,Y)$:
\begin{equation} \label{eq:MDL_ineq}
		h=p(0000)\left(\frac{1}{2}-\delta\right)^2 - \left( p(0101)+p(1010)+p(0011)\right)\left(\frac{1}{2}+\delta\right)^2 \leq 0\,.
\end{equation}
Assuming a non-signalling adversary classically correlated with the input source and the device, the amount of violation $h$ is then related to an upper bound on $p(A,B|X,Y)$, from which a lower bound on the min-entropy of outputs $AB$ is established, and consequently the extractable entropy\footnote{The extractable entropy is actually not computed in \cite{ramanathan2018practical}, where the chosen figure of merit is the number of $\varepsilon$-perfect random bits at the end of the protocol, meaning after randomness extraction. In Appendix \ref{app:boundsHminproofs}, we use their approach to obtain an expression for the extractable entropy that we can use to compare with our results.} after $n$ uses of the device.  Here we will lower bound this entropy by directly using our approach on noisy Hardy correlations, i.e., mixtures of behaviours that maximise the violation $h$ and white noise, 
\begin{equation}\label{eq:noisyHardy}
	p_{\text{Hardy}, w}(A,B,X,Y)=(1-w) p^*_{\text{Hardy}}(A,B,X,Y)+ w\frac{1}{16}\,.
\end{equation}
where we assumed $p(X,Y)=1/4$. The ideal correlations $p^{*}_{\rm Hardy}(A,B,X,Y)$ are obtained by measuring the two-qubit state $\ket{\phi_{\theta}}=\frac{1}{\sqrt{1+\cos^2\theta}}\left(\cos\theta\left(\ket{01}+\ket{10}\right) + \sin\theta\ket{11}\right)$ in the basis $\{\sin\theta\ket{0} - \cos\theta\ket{1}, \cos\theta\ket{0} + \sin\theta\ket{1}\}$ when $X,Y=0$ and $\{\ket{0}, \ket{1}\}$ if $X,Y=1$,  for $\theta=\arccos\sqrt{\frac{\sqrt{5}-1}{2}}$. Figure~\ref{fig:RA_Qproof} shows the extractable entropy rates obtained from $p_{\text{Hardy}, w}(A,B,X,Y)$ for different white noise and SV parameters with security parameter $\varepsilon=2^{-32}$.

\begin{figure*}[t]
  \centering
  \begin{subfigure}[t]{0.49\linewidth}
      \centering
      \includegraphics[width=\linewidth]{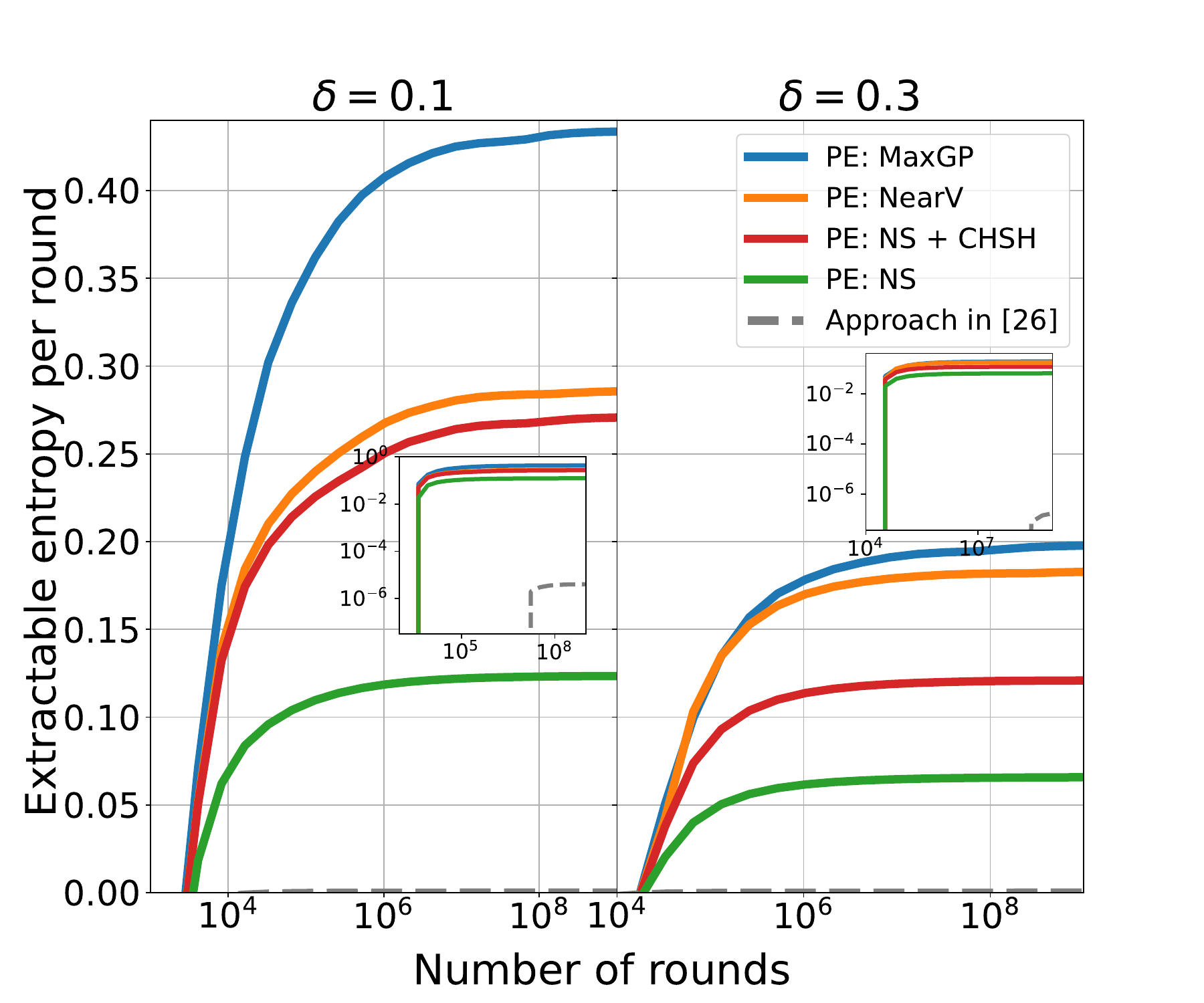}
      \caption{0.1\% white noise}
      \label{fig:RA_Qproof_bias_d1}
  \end{subfigure}
  \hfill
  \begin{subfigure}[t]{0.49\linewidth}
      \centering
      \includegraphics[width=\linewidth]{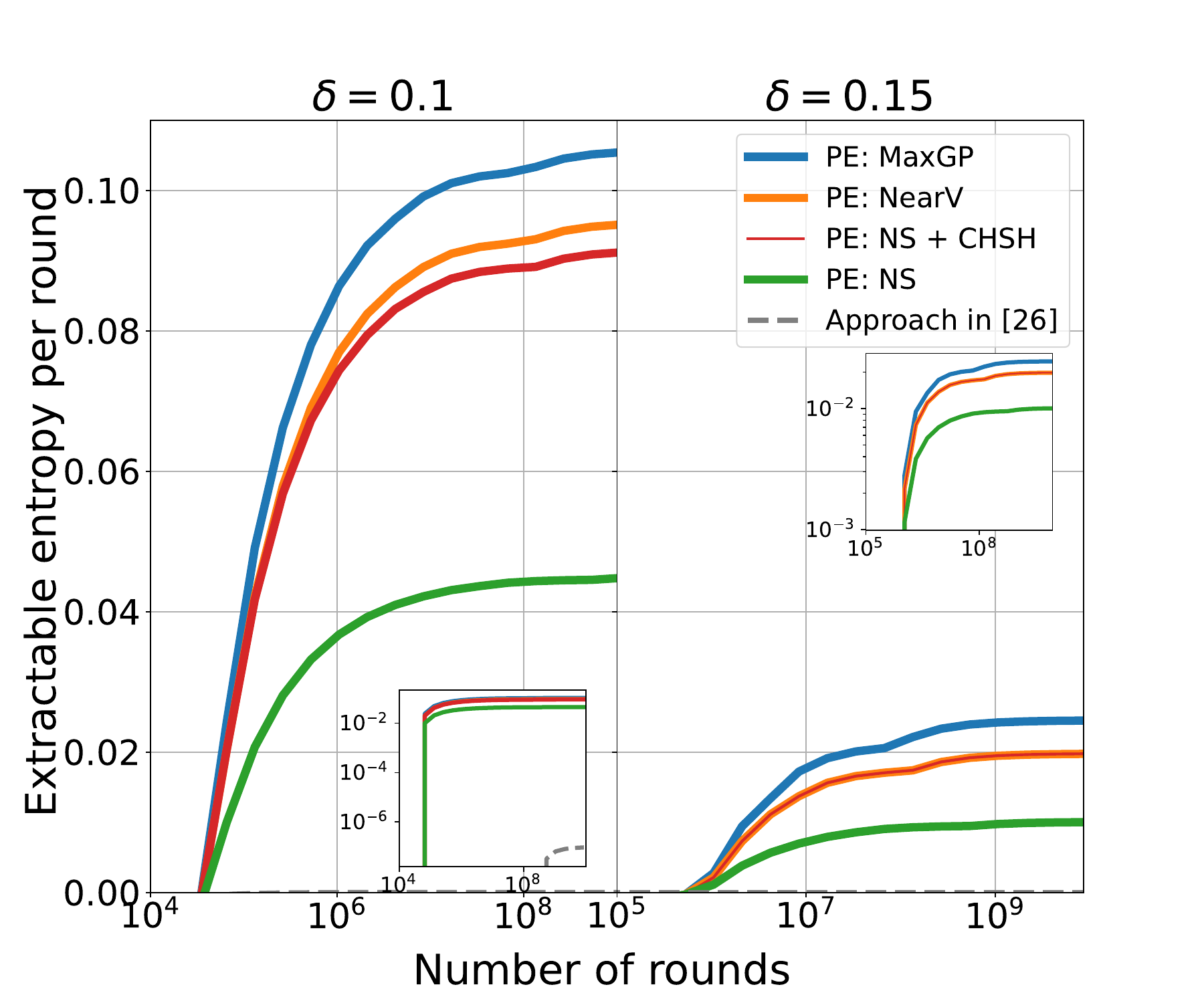}
      \caption{5\% white noise}
      \label{fig:RA_Qproof_bias_d3}
  \end{subfigure}
  \caption{\textbf{Randomness amplification of a SV-source.} Extractable entropy rates from noisy Hardy correlations \eqref{eq:noisyHardy} for different number of rounds: (a)  0.1\% white noise with SV parameter $\delta=0.1$ and $\delta=0.3$; and (b) 5\% white noise with SV parameter $\delta=0.1$ and $\delta=0.15$.  Lines corresponding to lower bounds computed with the PEF approach using $\poly_Q$ obtained through: (orange) 3 iterations of NearV (Algorithm~\ref{alg:closestEP} with $m=15$); (blue) 3 iterations of MaxGP (Algorithm~\ref{alg:gp}); (green) $\poly_Q=\ns$; (red) $\poly_Q=\ns \cap \textrm{CHSH} \leq 2\sqrt 2$.  The grey dashed line represents the extractable entropy using the approach in \cite{ramanathan2018practical} against NS adversaries (see Appendix~\ref{app:extractEnt_other}).
  The figure is produced using level 2 of the NPA hierarchy.}
  \label{fig:RA_Qproof}
\end{figure*}

\section{Discussion and conclusion}\label{sec:conclusion}

Our approach yields a significant improvement in the rates of extractable entropy across all scenarios we have examined. Additionally, when compared to analyses based on the original EAT \cite{Dupuis:2020EAT} or the Azuma-Hoeffding inequality as in \cite{pironio2013security}, our method requires far fewer device uses to achieve positive rates.

As discussed, a key advantage of PE is its better penalty factor---scaling as $O(\frac{1}{n})$, versus $O(\frac{1}{\sqrt{n}})$ for other approaches---compensating for suboptimal single round bounds. Another interesting observation is the distinct asymptotic behaviour of the three approaches and how this impacts their performance. In the PE approach, extractable entropy rates converge to the conditional Shannon entropy. In contrast, the method in \cite{pironio2013security} yields convergence to the conditional min-entropy, which is a lower bound on the conditional Shannon entropy. This difference contributes to the superior performance of PE-based methods unless the polytope approximation used is too coarse (as is the case in the tripartite scenario analysed).
On the other hand, entropy extracted using the EAT converges to the conditional von Neumann entropy, which is generally lower than the conditional Shannon entropy. However, in the large $n$ regime and for correlations with low noise levels, the EAT approach sometimes outperforms the PE method. We believe this is primarily due to a smaller set of allowed behaviours (defined directly by the NPA hierarchy) being considered in the EAT framework.

Regarding the scalability of our approach, it is important to highlight that both algorithms fundamentally rely on the ability to perform vertex enumeration---a computationally intensive task that constitutes a clear bottleneck. For example, computing the vertices of the non-signalling polytope $\ns$ in Bell scenarios involving more than three parties becomes infeasible with standard computational resources.  As a result, applying our methods to more complex scenarios would require initialising the algorithms with significantly coarser approximations of the quantum set---such as polytopes that may include signalling behaviours. This possibility could be explored if these more complex scenarios become relevant at the experimental level. 

Going beyond the framework analysed here, we could immediately use our construction to prove security against quantum side information, in specific scenarios, by following the method proposed in \cite{zhangQPE}. Further, it would be interesting to investigate how our approach can be used for the security analysis of semi-DI-QRNG/QKD protocols where the set of allowed behaviours can be approximated by semidefinite constraints; refer to \cite{Tavakoli_SDP_review} for relevant examples of such protocols.

\paragraph{Acknowledgements}
The authors thank Cameron Foreman, Erik Woodhead and Victoria Wright for helpful discussions and feedback on an earlier version of the manuscript.

\bibliographystyle{quantum}
\bibliography{mybibliography}

\begin{thebibliography}{10}

\bibitem{hayashi2016more}
Masahito Hayashi and Toyohiro Tsurumaru.
\newblock ``More efficient privacy amplification with less random seeds via
  dual universal hash function''.
\newblock \href{https://dx.doi.org/10.1109/TIT.2016.2526018}{IEEE Transactions
  on Information Theory {\bf 62}, 2213--2232}~(2016).

\bibitem{foreman2025cryptomite}
Cameron Foreman, Richie Yeung, Alec Edgington, and Florian~J. Curchod.
\newblock ``Cryptomite: A versatile and user-friendly library of randomness
  extractors''.
\newblock \href{https://dx.doi.org/10.22331/q-2025-01-08-1584}{Quantum {\bf 9},
  1584}~(2025).

\bibitem{pironio2010random}
Stefano Pironio, Antonio Ac{\'\i}n, Serge Massar, A~Boyer de~La~Giroday,
  Dzmitry~N. Matsukevich, Peter Maunz, Steven Olmschenk, David Hayes, Le~Luo,
  Timothy~A. Manning, et~al.
\newblock ``Random numbers certified by {B}ell’s theorem''.
\newblock \href{https://dx.doi.org/https://doi.org/10.1038/nature09008}{Nature
  {\bf 464}, 1021--1024}~(2010).

\bibitem{pironio2013security}
Stefano Pironio and Serge Massar.
\newblock ``Security of practical private randomness generation''.
\newblock
  \href{https://dx.doi.org/https://doi.org/10.1103/PhysRevA.87.012336}{Physical
  Review A {\bf 87}, 012336}~(2013).

\bibitem{fehr2013security}
Serge Fehr, Ran Gelles, and Christian Schaffner.
\newblock ``Security and composability of randomness expansion from {B}ell
  inequalities''.
\newblock
  \href{https://dx.doi.org/https://doi.org/10.1103/PhysRevA.87.012335}{Physical
  Review A {\bf 87}, 012335}~(2013).

\bibitem{Dupuis:2020EAT}
Fr{\'e}d{\'e}ric Dupuis, Omar Fawzi, and Renato Renner.
\newblock ``Entropy accumulation''.
\newblock
  \href{https://dx.doi.org/https://doi.org/10.1007/s00220-020-03839-5}{Communications
  in Mathematical Physics {\bf 379}, 867--913}~(2020).

\bibitem{zhang2018certifying}
Yanbao Zhang, Emanuel Knill, and Peter Bierhorst.
\newblock ``Certifying quantum randomness by probability estimation''.
\newblock
  \href{https://dx.doi.org/https://doi.org/10.1103/PhysRevA.98.040304}{Physical
  Review A {\bf 98}, 040304}~(2018).

\bibitem{knill2020generation}
Emanuel Knill, Yanbao Zhang, and Peter Bierhorst.
\newblock ``Generation of quantum randomness by probability estimation with
  classical side information''.
\newblock
  \href{https://dx.doi.org/https://doi.org/10.1103/PhysRevResearch.2.033465}{Physical
  Review Research {\bf 2}, 033465}~(2020).

\bibitem{shalmPEFinblocks}
Lynden~K. Shalm, Yanbao Zhang, Joshua~C. Bienfang, Collin Schlager, Martin~J.
  Stevens, Michael~D. Mazurek, Carlos Abell{\'a}n, Waldimar Amaya, Morgan~W.
  Mitchell, Mohammad~A. Alhejji, Honghao Fu, Joel Ornstein, Richard~P. Mirin,
  Sae~Woo Nam, and Emanuel Knill.
\newblock ``Device-independent randomness expansion with entangled photons''.
\newblock
  \href{https://dx.doi.org/https://doi.org/10.1038/s41567-020-01153-4}{Nature
  Physics {\bf 17}, 452--456}~(2021).

\bibitem{bierhorst2020tsirelson}
Peter Bierhorst and Yanbao Zhang.
\newblock ``Tsirelson polytopes and randomness generation''.
\newblock \href{https://dx.doi.org/10.1088/1367-2630/aba30d}{New Journal of
  Physics {\bf 22}, 083036}~(2020).

\bibitem{acin2012randomness}
Antonio Ac{\'\i}n, Serge Massar, and Stefano Pironio.
\newblock ``Randomness versus nonlocality and entanglement''.
\newblock
  \href{https://dx.doi.org/https://doi.org/10.1103/PhysRevLett.108.100402}{Physical
  Review Letters {\bf 108}, 100402}~(2012).

\bibitem{Tavakoli_SDP_review}
Armin Tavakoli, Alejandro Pozas-Kerstjens, Peter Brown, and Mateus Ara\'ujo.
\newblock ``Semidefinite programming relaxations for quantum correlations''.
\newblock \href{https://dx.doi.org/10.1103/RevModPhys.96.045006}{Reviews of
  Modern Physics {\bf 96}, 045006}~(2024).

\bibitem{popescu1994quantum}
Sandu Popescu and Daniel Rohrlich.
\newblock ``Quantum nonlocality as an axiom''.
\newblock
  \href{https://dx.doi.org/https://doi.org/10.1007/BF02058098}{Foundations of
  Physics {\bf 24}, 379--385}~(1994).

\bibitem{boundedStorageKonigTerhal}
Robert~T. Konig and Barbara~M. Terhal.
\newblock ``The bounded-storage model in the presence of a quantum adversary''.
\newblock \href{https://dx.doi.org/10.1109/TIT.2007.913245}{IEEE Transactions
  on Information Theory {\bf 54}, 749--762}~(2008).

\bibitem{RenWol04a}
Renato Renner and Stefan Wolf.
\newblock ``Smooth {R}enyi entropy and applications''.
\newblock In IEEE International Symposium on Information Theory --- ISIT 2004.
\newblock
  \href{https://dx.doi.org/https://doi.org/10.1109/ISIT.2004.1365269}{Page
  233}.
\newblock IEEE~(2004).

\bibitem{tsirelson1993some}
Boris~S. Tsirelson.
\newblock ``Some results and problems on quantum {B}ell-type inequalities''.
\newblock Hadronic Journal Supplement {\bf 8}, 329--345~(1993).

\bibitem{navascues2007bounding}
Miguel Navascu{\'e}s, Stefano Pironio, and Antonio Ac{\'\i}n.
\newblock ``Bounding the set of quantum correlations''.
\newblock
  \href{https://dx.doi.org/https://doi.org/10.1103/PhysRevLett.98.010401}{Physical
  Review Letters {\bf 98}, 010401}~(2007).

\bibitem{navascues2008convergent}
Miguel Navascu{\'e}s, Stefano Pironio, and Antonio Ac{\'\i}n.
\newblock ``A convergent hierarchy of semidefinite programs characterizing the
  set of quantum correlations''.
\newblock \href{https://dx.doi.org/10.1088/1367-2630/10/7/073013}{New Journal
  of Physics {\bf 10}, 073013}~(2008).

\bibitem{BellReview2014}
Nicolas Brunner, Daniel Cavalcanti, Stefano Pironio, Valerio Scarani, and
  Stephanie Wehner.
\newblock ``Bell nonlocality''.
\newblock \href{https://dx.doi.org/10.1103/RevModPhys.86.419}{Reviews of Modern
  Physics {\bf 86}, 419--478}~(2014).

\bibitem{nieto2014using}
Olmo Nieto-Silleras, Stefano Pironio, and Jonathan Silman.
\newblock ``Using complete measurement statistics for optimal
  device-independent randomness evaluation''.
\newblock \href{https://dx.doi.org/10.1088/1367-2630/16/1/013035}{New Journal
  of Physics {\bf 16}, 013035}~(2014).

\bibitem{bancal2014more}
Jean-Daniel Bancal, Lana Sheridan, and Valerio Scarani.
\newblock ``More randomness from the same data''.
\newblock \href{https://dx.doi.org/10.1088/1367-2630/16/3/033011}{New Journal
  of Physics {\bf 16}, 033011}~(2014).

\bibitem{clauser1969proposed}
John~F. Clauser, Michael~A. Horne, Abner Shimony, and Richard~A. Holt.
\newblock ``Proposed experiment to test local hidden-variable theories''.
\newblock
  \href{https://dx.doi.org/https://doi.org/10.1103/PhysRevLett.23.880}{Physical
  Review Letters {\bf 23}, 880}~(1969).

\bibitem{rosenfeld2017event}
Wenjamin Rosenfeld, Daniel Burchardt, Robert Garthoff, Kai Redeker, Norbert
  Ortegel, Markus Rau, and Harald Weinfurter.
\newblock ``Event-ready {B}ell test using entangled atoms simultaneously
  closing detection and locality loopholes''.
\newblock
  \href{https://dx.doi.org/https://doi.org/10.1103/PhysRevLett.119.010402}{Physical
  Review Letters {\bf 119}, 010402}~(2017).

\bibitem{zhang2020experimental}
Yanbao Zhang, Lynden~K. Shalm, Joshua~C. Bienfang, Martin~J. Stevens,
  Michael~D. Mazurek, Sae~Woo Nam, Carlos Abell{\'a}n, Waldimar Amaya,
  Morgan~W. Mitchell, Honghao Fu, et~al.
\newblock ``Experimental low-latency device-independent quantum randomness''.
\newblock
  \href{https://dx.doi.org/https://doi.org/10.1103/PhysRevLett.124.010505}{Physical
  Review Letters {\bf 124}, 010505}~(2020).

\bibitem{nieto2018severalBell}
Olmo Nieto-Silleras, Cédric Bamps, Jonathan Silman, and Stefano Pironio.
\newblock ``Device-independent randomness generation from several {B}ell
  estimators''.
\newblock \href{https://dx.doi.org/10.1088/1367-2630/aaaa06}{New Journal of
  Physics {\bf 20}, 023049}~(2018).

\bibitem{ramanathan2018practical}
Ravishankar Ramanathan, Micha{\l} Horodecki, Hammad Anwer, Stefano Pironio,
  Karol Horodecki, Marcus Gr{\"u}nfeld, Sadiq Muhammad, Mohamed Bourennane, and
  Pawe{\l} Horodecki.
\newblock ``Practical no-signalling proof randomness amplification using
  {H}ardy paradoxes and its experimental implementation''.
\newblock
  \href{https://dx.doi.org/https://doi.org/10.48550/arXiv.1810.11648}{ar{X}iv
  preprint ar{X}iv:1810.11648}~(2018).

\bibitem{arnon2018practical}
Rotem Arnon-Friedman, Fr{\'e}d{\'e}ric Dupuis, Omar Fawzi, Renato Renner, and
  Thomas Vidick.
\newblock ``Practical device-independent quantum cryptography via entropy
  accumulation''.
\newblock
  \href{https://dx.doi.org/https://doi.org/10.1038/s41467-017-02307-4}{Nature
  Communications {\bf 9}, 459}~(2018).

\bibitem{QuantinuumH1}
``Quantinuum {H}1-1. https://www.quantinuum.com/, {M}ay, 2023.''.

\bibitem{lin2018device}
Pei-Sheng Lin, Denis Rosset, Yanbao Zhang, Jean-Daniel Bancal, and Yeong-Cherng
  Liang.
\newblock ``Device-independent point estimation from finite data and its
  application to device-independent property estimation''.
\newblock
  \href{https://dx.doi.org/https://doi.org/10.1103/PhysRevA.97.032309}{Physical
  Review A {\bf 97}, 032309}~(2018).

\bibitem{pironio2005lifting}
Stefano Pironio.
\newblock ``Lifting {B}ell inequalities''.
\newblock \href{https://dx.doi.org/https://doi.org/10.1063/1.1928727}{Journal
  of Mathematical Physics{\bf 46}}~(2005).

\bibitem{Woodhead2018randomnessversus}
Erik Woodhead, Boris Bourdoncle, and Antonio Ac{\'{i}}n.
\newblock ``Randomness versus nonlocality in the {M}ermin-{B}ell experiment
  with three parties''.
\newblock \href{https://dx.doi.org/10.22331/q-2018-08-17-82}{{Quantum} {\bf 2},
  82}~(2018).

\bibitem{santha1986generating}
Miklos Santha and Umesh~V. Vazirani.
\newblock ``Generating quasi-random sequences from semi-random sources''.
\newblock
  \href{https://dx.doi.org/https://doi.org/10.1016/0022-0000(86)90044-9}{Journal
  of Computer and System Sciences {\bf 33}, 75--87}~(1986).

\bibitem{HardyParadox}
Lucien Hardy.
\newblock ``Quantum mechanics, local realistic theories, and
  {L}orentz-invariant realistic theories''.
\newblock \href{https://dx.doi.org/10.1103/PhysRevLett.68.2981}{Physical Review
  Letters {\bf 68}, 2981--2984}~(1992).

\bibitem{putz2016measurement}
Gilles P{\"u}tz and Nicolas Gisin.
\newblock ``Measurement dependent locality''.
\newblock
  \href{https://dx.doi.org/https://doi.org/10.1088/1367-2630/18/5/055006}{New
  Journal of Physics {\bf 18}, 055006}~(2016).

\bibitem{zhangQPE}
Yanbao Zhang, Honghao Fu, and Emanuel Knill.
\newblock ``Efficient randomness certification by quantum probability
  estimation''.
\newblock \href{https://dx.doi.org/10.1103/PhysRevResearch.2.013016}{Physical
  Review Research {\bf 2}, 013016}~(2020).

\bibitem{van2019correlations}
Thomas Van~Himbeeck and Stefano Pironio.
\newblock ``Correlations and randomness generation based on energy
  constraints''.
\newblock
  \href{https://dx.doi.org/https://doi.org/10.48550/arXiv.1905.09117}{ar{X}iv
  preprint ar{X}iv:1905.09117}~(2019).

\bibitem{hoeffding1994probability}
Wassily Hoeffding.
\newblock ``Probability inequalities for sums of bounded random variables''.
\newblock \href{https://dx.doi.org/10.1080/01621459.1963.10500830}{The
  collected works of Wassily HoeffdingPages 409--426}~(1994).

\bibitem{genEAT}
Tony Metger, Omar Fawzi, David Sutter, and Renato Renner.
\newblock ``Generalised entropy accumulation''.
\newblock In 2022 IEEE 63rd Annual Symposium on Foundations of Computer Science
  (FOCS).
\newblock \href{https://dx.doi.org/10.1109/FOCS54457.2022.00085}{Pages
  844--850}.
\newblock ~(2022).

\bibitem{brown2021device}
Peter Brown, Hamza Fawzi, and Omar Fawzi.
\newblock ``Device-independent lower bounds on the conditional von {N}eumann
  entropy''.
\newblock
  \href{https://dx.doi.org/https://doi.org/10.22331/q-2024-08-27-1445}{Quantum
  {\bf 8}, 1445}~(2024).

\bibitem{brandaoRealisticAmpFew}
Fernando G. S.~L. Brand{\~a}o, Ravishankar Ramanathan, Andrzej Grudka, Karol
  Horodecki, Micha{\l} Horodecki, Pawe{\l} Horodecki, Tomasz Szarek, and Hanna
  Wojew{\'o}dka.
\newblock ``Realistic noise-tolerant randomness amplification using finite
  number of devices''.
\newblock \href{https://dx.doi.org/https://doi.org/10.1038/ncomms11345}{Nature
  Communications {\bf 7}, 11345}~(2016).

\end{thebibliography}

\newpage
\appendix
\appendixpage

\section{Bounds on smooth min-entropies}
\label{app:boundsHminproofs}
	
	A useful tool to bound the total variation distance between two probability distributions of $X$ is the notion of \emph{sub-probability distribution $\tilde \nu(X)$}. This is a sub-normalised probability distribution, i.e $\tilde \nu(x)\geq 0, \forall x$, and the weight of $\tilde{\nu}$ satisfies $\omega(\tilde{\nu}):=\sum_x \tilde{\nu}(x)\leq 1$. 
	The following lemma is the generalisation of Lemma 1 in \cite{knill2020generation}.
	
	\begin{lemma}\label{th:subprobTV}
		Let $\nu$ and $\nu'$ be probability distributions of $X$. Then $\|\nu-\nu'\|\leq \epsilon$  if and only if there exists a sub-probability distribution $\tilde\nu$ of $X$, with weight $\omega\geq 1-\epsilon$, satisfying $\tilde{\nu}\leq\nu$ and $\tilde{\nu}\leq\nu'$.\footnote{We use the notation $\nu\leq\nu'$ when $\nu(x)\leq\nu'(x)$ for all $x$.}
	\end{lemma}
	\begin{proof}
		If there is $\tilde \nu$ such that $\tilde{\nu}\leq\nu$ and $\tilde{\nu}\leq\nu'$, with weight $\omega\geq 1-\epsilon$, we can directly check that $\|\nu-\nu'\|\leq \epsilon$. Start by expressing the total variation as 
		\begin{equation}
			\|\nu-\nu'\|=\sum_x\left(\nu(x)-\nu'(x)\right)\Theta(\nu(x)-\nu'(x))\,,
		\end{equation}
		where $\Theta(y)=1$ if $y\geq 0$ and $\Theta(y)=0$ otherwise. Then,
		\begin{equation}
			\begin{split}
				\|\nu-\nu'\|&=\sum_x\left(\nu(x)-\nu'(x)\right)\Theta(\nu(x)-\nu'(x)) \\  & \leq \sum_x\left(\nu(x)-\tilde \nu(x)\right)\leq 1-\omega\leq \epsilon\,.
			\end{split}
		\end{equation}
		
		In the other direction, if $\|\nu-\nu'\|\leq \epsilon$, we can define a sub-probability distribution as $\tilde{\nu}(x):=\min\{\nu(x),\nu'(x)\}$. Then $\tilde{\nu}\leq\nu$ and $\tilde{\nu}\leq\nu'$ as required. We are left with proving that the weight of $\tilde \nu$ is $\omega\geq 1-\epsilon$, which follows from:
		\begin{equation}
			\begin{split}
				\omega=&\sum_x\min\{\nu(x),\nu'(x)\}\\
				=&\sum_x\nu(x)\Theta\left(\nu'(x)-\nu(x)\right)+\sum_x\nu'(x)\Theta\left(\nu(x)-\nu'(x)\right)\\
				=&\sum_x\nu(x)\Theta\left(\nu'(x)-\nu(x)\right)+\left(1-\sum_x\nu'(x)\Theta\left(\nu'(x)-\nu(x)\right)\right)\\
				=&1-\|\nu-\nu'\|\geq 1-\epsilon\,.
			\end{split}
		\end{equation}
		In the third line we have used the identity
		\begin{equation} 
			1=\sum_x\nu'(x)=\sum_x\nu'(x)\left[\Theta\left(\nu'(x)-\nu(x)\right)+\Theta\left(\nu(x)-\nu'(x)\right)\right]\,.
		\end{equation}
	\end{proof}
	
	We can immediately use this result to prove the following relation between high-confidence bounds on the min-entropy and the smooth min-entropy. We are using here generic variables $C$, $D$ and $Z$, while in the main text we have explicitly represented them as sequences.
	
	\begin{theorem}[Theorem \ref{th:confbound2minent} in the main text]
    \label{th:confbound2Hmin_app}
		If, for all $e$, $\prob{  H_{\min}(D|Z,e)\geq h}\geq 1-\epsilon$ then $H_{\min}^{\epsilon}(D|Z,E)\geq h$.
	\end{theorem}
	\begin{proof}
		Assume that $DZE$ is distributed according to $\mu$. Define the sub-distribution $\tilde \nu$,
		\begin{equation}
			\tilde \nu(D,Z,E)=\mu(D,Z,E)\Theta\left(2^{-h}-\max_{d'}\mu(d'|Z,E)\right)\,,
		\end{equation}
		where $\Theta(x)=1$ if $x\geq 0$ and $\Theta(x)=0$ otherwise. The weight of $\tilde \nu$ is
		\begin{equation}
			\begin{aligned}
				\omega&=\sum_{d,z,e}\mu(z,e)\mu(d|z,e)\Theta\left(2^{-h}-\max_{d'}\mu(d'|z,e)\right)\\
                &=\sum_{z,e}\mu(z,e)\Theta\left(2^{-h}-\max_{d'}\mu(d'|z,e)\right)\geq 1-\epsilon\,.
			\end{aligned}
		\end{equation}
		Now define $\nu$ as
		\begin{equation}
			\nu(D,Z,E) = 
			\begin{cases}
				\mu(D,Z,E) & \text{if } \max_{d'}\mu(d'|Z,E)\leq 2^{-h}\\
				\frac{1}{|D|}{\mu(Z,E)} & \text{otherwise.}
			\end{cases}
		\end{equation}
		Since $2^{-h} \geq 1/|D|$
        it follows that $\max_d\nu(d|z,e)\leq 2^{-h}$, i.e. $H_{\min ,\nu}(D|z,e)\geq h,\forall z,e$. Clearly $\tilde \nu\leq \mu$ and $\tilde \nu \leq \nu$, and by Lemma \ref{th:subprobTV}, $\|\mu-\nu\|\leq \epsilon$. Then  $H_{\min,\mu}^{\epsilon}(D|Z,E)\geq h$.
	\end{proof}
	
	The following average notion of conditional smooth min-entropy  comes in handy for some of the proofs in this paper.
	\begin{definition}\label{def:smoothAvgHmin}
		Let random variables $DZE$ be jointly distributed according to $\mu$.  Then $D$ has \emph{$\epsilon$-smooth average ZE-conditional min-entropy} $H_{\min,\mu}^{\epsilon,\text{\emph{avg}}}(D|Z,E)\geq h$ if there exists a distribution $\nu$ such that
		\begin{enumerate}
			\item $\|\mu-\nu\|\leq \epsilon$,
			\item $\mathbb{E}(\max_d\nu(d|ZE))\leq 2^{-h}$.
		\end{enumerate}
	\end{definition}
	
	Moreover, it is convenient to find a mapping between the average conditional min-entropy and the (worst-case) version in Definition~\ref{def:smoothHmin}. While the latter clearly provides a lower bound on the average version, i.e., if $H_{\min,\mu}^{\epsilon}(D|Z,E)\geq h$ then $H_{\min,\mu}^{\epsilon,\text{avg}}(D|Z,E)\geq h$, a relationship in the opposite direction is not as straightforward. The following theorem establishes that result and is adapted from Lemma 5 in~\cite{knill2020generation} to our (more general) definition of conditional smooth min-entropy in Definition~\ref{def:smoothHmin}.
	
	\begin{lemma}\label{th:avg2worstsmoothH}
		Consider positive constants $p>0$ and $\delta>0$. If $H_{\min,\mu}^{\epsilon,\text{\emph{avg}}}(D|Z,E)\geq -\log(p\delta)$, then $H_{\min,\mu}^{\epsilon+\delta}(D|Z,E)\geq -\log p$.
	\end{lemma}
	\begin{proof}
		By definition, there exists a distribution $\nu$ of $DZE$ such that $\mathbb{E}(\max_d\nu(d|Z,E))\leq p\delta$ and $\|\mu-\nu\|\leq\epsilon$. Consider the random variable $M(Z,E)=\max_d\nu(d|Z,E)$. Applying the Markov inequality, we get
		\begin{equation}\label{eq:badevents}
			\Pr(M\geq p)\leq \delta.
		\end{equation}
		
		Define 
		\begin{equation}
			\eta(D,Z,E) = 
			\begin{cases}
				\nu(D,Z,E) & \text{if } \max_{d'}\nu(d'|Z,E)\leq p\\
				\frac{1}{|D|}{\nu(Z,E)}, & \text{otherwise.}
			\end{cases}
		\end{equation}
		
		We will now prove that there exists a sub-distribution $\tilde \nu$ with weight $\omega\geq 1-\delta$ such that $\tilde \nu\leq \nu$ and $\tilde \nu\leq \eta$. Define $\tilde{\nu}(D,Z,E)=\nu(D,Z,E)\Theta(p-\max_{d'}\nu(d'|Z,E))$. It is clear that $\tilde \nu\leq \nu$ and $\tilde \nu\leq \eta$ follows from the fact that $p\geq 1/|D|$. Furthermore, the weight of $\tilde{\nu}$ is
		\begin{equation}
			\begin{aligned}
				\omega=\sum_{d,z,e}\tilde\nu(d,z,e)=&\sum_{d,z,e}\nu(d|z,e)\nu(z,e)\Theta\left(p-\max_{d'}\nu(d'|z,e)\right) \\ = &\sum_{z,e}\nu(z,e)\Theta\left(p-\max_{d'} \nu(d'|z,e)\right)\geq 1-\delta\,,
			\end{aligned}
		\end{equation}
	where the last equality comes from the fact that $\sum_d \nu(d|ze)=1$ while in the final inequality we used the bound~\eqref{eq:badevents}. Then, from Lemma~\ref{th:subprobTV}, $\|\nu-\eta\|\leq \delta$. Using the triangle inequality, we get 
		\begin{equation}
			\|\mu-\eta\|\leq \|\mu-\nu\|+\|\nu-\eta\|\leq \epsilon+\delta\,.
		\end{equation}
		Since $\eta(d|z,e)\leq p, \forall d,z,e$, we finally conclude that $H_{\min,\mu}^{\epsilon+\delta}(D|Z,E)\geq -\log p$.
	\end{proof}
	
	\begin{theorem}[Theorem \ref{thm:extractable_E} in the main text]
		If, for all $e$, $\prob{H_{\min}(D|Z,e)\geq W}\geq 1-\kappa$ and $p_{Acc}=\Pr(W\geq t),$ then the \emph{extractable entropy of $D$} is lower-bounded by
		\begin{equation}
        \label{eq:app:extractable_entropy}
			H_{\min}^{\varepsilon/p_{Acc}}(D|Z,E,\verb|Accept|)\geq t +\frac{\log \kappa}{\beta}+\log(\varepsilon-\kappa) \,.
		\end{equation}
	\end{theorem}

	\begin{proof}
		The proof follows from Lemma 18 in \cite{knill2020generation}, which states that under these conditions the average conditional smooth min-entropy satisfies
		\begin{equation}
			H_{\min}^{\kappa/p_{Acc},\text{avg}}\left(D|Z,E,W\geq t\right)\geq t +\frac{\log \kappa}{\beta}+\log(p_{Acc}).
		\end{equation}
		Then, 
		\begin{equation}
			H_{\min}^{\kappa/p_{Acc},\text{avg}}\left(D|Z,E,W\geq t\right)\geq t +\frac{\log \kappa}{\beta}+\log(\varepsilon-\kappa)-\log\left(\frac{\varepsilon-\kappa}{p_{Acc}}\right)
		\end{equation}
		to which we apply Lemma \ref{th:avg2worstsmoothH} and finally obtain \eqref{eq:app:extractable_entropy}.
	\end{proof}

	\section{PEFs and min-trade-off functions}
	\label{app:PEFmintradeoff}

	Min-trade-off functions were introduced in the entropy accumulation framework \cite{Dupuis:2020EAT} and provide a convenient tool to lower-bound single round entropies.  We revisit here a connection between PEFs and min-trade-off functions established in \cite{knill2020generation}, which will later allow us to compare the extractable entropies calculated using both methods. 
	
	Consider the \emph{average conditional Shannon entropy} defined as 
	\begin{equation}\label{eq:ShannonEnt}
		H(D|Z,E)=\sum_e \mu(e)H(D|Z,e) =-\sum_{dze}\mu(e)\mu(dz|e)\log \mu(d|ze)\,,
	\end{equation}
	and an affine function $f(\mu)=\sum_{cz}K(c,z)\mu(cz)$ which satisfies
	\begin{equation}\label{def:mintradeoff}
		f(\mu)\leq H(D|Z,E)\,,\qquad \forall \mu\in\phyM\,,
	\end{equation}
where $\phyM$ is the set of allowed behaviours and $f$ is a \emph{min-trade-off function} for $H(D|Z,E)$ . The function $f$ can be seen as the expected value of a \emph{min-trade-off estimator} $K=K(C,Z)$
	\begin{equation}
		\mathbb{E}(K(C,Z))=f(\mu)\,,
	\end{equation}
	where $\mu=\mu(C,Z)$.  As a result, the optimal $K$ for a given  behaviour $p$ is the solution of the optimisation  \cite{van2019correlations}, 
	\begin{equation}\label{eq:mintradeoffopt}
		\begin{aligned}
			\max_K\quad & \sum_{cz} p(cz) K(c,z)\\
			\textrm{s.t.} \quad &  \sum_{cz}\mu(c,z)K(c,z)\leq \sum_{dz}\mu(d,z)(-\log \mu(d|z))\,,\quad \forall \mu \in \phyM\,.\\
		\end{aligned}
	\end{equation} 
	There is clear a resemblance between \eqref{eq:mintradeoffopt} and the single round PEFs optimisation \eqref{eq:PEFoptnobeta}. In fact, we will see next that every PEF can be used to define a min-trade-off function \eqref{def:mintradeoff}.

	\begin{theorem}
		Consider a PEF $F(C,Z)$ with power $\beta$. Then $\frac{1}{\beta}\log F$ is a (suboptimal) min-trade-off estimator for the average conditional Shannon entropy $H(D|Z,e)$ \eqref{eq:ShannonEnt}.
	\end{theorem}
	\begin{proof}
		Applying the log function to condition \ref{def:PEFcons} in the definition of PEF, we have 
		\begin{equation}\label{eq:PEFconstK}
			\log\left[\mathbb{E}\left(F(C,Z)\mu(D|Z)^\beta\right)\right]\leq 0\,,\qquad \forall \mu \in \phyM\,.
		\end{equation}
		Since $\log$ is a concave function, we can use Jensen's inequality to write \eqref{eq:PEFconstK} as 
		\begin{equation}
			\mathbb{E}\left[\log\left(F(C,Z)\mu(D|Z)^\beta\right)\right]+\Delta\leq 0\,,\qquad \forall \mu \in \phyM\,,
		\end{equation}
		where $\Delta:=\log \left[\mathbb{E}(F(C,Z)\mu(D|Z)^\beta)\right]- \mathbb{E}\left[ \log (F(C,Z)\mu(D|Z)^\beta)\right]\geq 0$.
		From this, we can finally prove that $\log F/\beta$ is a min-trade-off estimator for the conditional Shannon entropy.
		\begin{equation}\label{eq:mintradeconst}
			\mathbb{E}\left[\frac{\log F(C,Z)}{\beta}\right]+\frac{\Delta}{\beta}\leq \mathbb{E}\left[-\log \mu(D|Z)\right] = H(D|Z)\,,\quad \forall \mu\in \phyM\,.
		\end{equation} 
	\end{proof}
	Note that due to the $\Delta/\beta$ gap, this estimator will often be suboptimal. In \cite{knill2020generation}, it is proven that in the asymptotic limit of very large $n$, $\beta\rightarrow 0$ and $\frac{1}{\beta}\log F$ tends to an optimal estimator of the average conditional Shannon entropy.

\section{Extractable entropy from alternative methods}
\label{app:extractEnt_other}
In this section, we give a brief description of how we applied the alternative methods to bound the extractable entropy in Section~\ref{sec:DIrandVerif}. Since the aim is to compare this with the performance of the PEF method coupled with our proposed polytope approximations, we took the off-the-shelf results which provide the fairest comparison. 

\subsection{Classical side information: from single round $P_{\rm guess}$ and the Azuma-Hoeffding concentration inequality}

The first method to successfully bound the extractable entropy for sequential processes in devices with memory and classical side information was introduced in \cite{pironio2010random,pironio2013security,fehr2013security}. Here we mainly follow the approach in  \cite{pironio2013security}, which derives a lower bound on the extractable entropy using a min-trade-off function for $H_{\min}(D|Z,E):=-\log(\max_{d,z,e}p(d|z,e))$ and the Azuma-Hoeffding concentration inequality \cite{hoeffding1994probability}. 

The first step is to find a function that lower bounds the min-entropy for every input $z$. Consider the generalisation of optimisation  \eqref{eq:GP_A_givenxbar} described in \cite{nieto2018severalBell},
\begin{align}
	\begin{split}	\label{eq:guessprob_opt}
		P_{\rm guess}(D|Z,\Lambda) = \max_{\{\tilde{\mu}_{\lambda}\}_{\lambda}}  & \sum_{d,z} \tilde{\mu}_{d,z}(d|z)\\
		\text{s.t. } & \sum_{d,z} \tilde{\mu}_{d,z}(C|Z) = p(C|Z)\,, \quad \forall \tilde{\mu}_{d,z} (C|Z)\in \widetilde{\phyM}\,.
	\end{split}
\end{align}

The dual problem can be written as
\begin{align}
	\begin{split}	\label{eq:dual_guessprob_opt}
	 	\min_{B}  & \sum_{c,z} p(c,z)B(c,z)\\
		\text{s.t. } &  \sum_{c,z} \mu(c,z)B(c,z) \geq \max_{d,z}\mu(d|z)\,, \quad \forall \mu \in \qs\,,
	\end{split}
\end{align}
where $\mu(z)=p(z)$. The optimal solution to this problem provides an estimator $B$, optimised for the typical behaviour $p$, that satisfies
\begin{equation}\label{eq:mintradeoffHmin}
	-\log \mathbb{E}_{\mu}(B)\leq -\log \max_{d,z}\mu(d|z), \qquad \forall \mu\in\qs\,,
\end{equation}
and therefore $h(\mu):=-\log \mathbb{E}_{\mu}(B)$ is a min-trade-off function for $H_{\min}(D|Z,E)$, as described in Appendix~\ref{app:PEFmintradeoff}.

Following the method described in \cite{pironio2013security}, we use the single round bound  \eqref{eq:mintradeoffHmin} combined with the Azuma-Hoeffding inequality to obtain the extractable entropy
\begin{equation}\label{eq:HminAZ}
	H_{\min}^{\varepsilon/p_{Acc}}(\mathbf{D}|\mathbf{Z},E,\verb|Accept|)\geq  nt-\gamma\left(2n\ln 1/\kappa\right)^{1/2} -\log 1/(\varepsilon-\kappa) \,,
\end{equation}
where $h(\mu)\geq t$ is the acceptance criterion and $\gamma=\max|B-\mathbb{E}_\mu(B)|\,, \forall \mu\in \qs$. We obtained the optimal estimator $B$ by solving \eqref{eq:dual_guessprob_opt} for the typical behaviours $p(C,Z)$ corresponding to noisy CHSH correlations in Section~\ref{sec:noisyCHSH} and accordingly chose $t=-\log \mathbb{E}_p(B)$.

\subsection{Quantum side information: EAT}

The entropy accumulation theorem (EAT) introduced in \cite{Dupuis:2020EAT} allows us to bound the conditional von Neumann entropy of outputs from a sequential process, and therefore include the general case where an adversary holds quantum side information. 
In order to treat the framework of the original EAT \cite{Dupuis:2020EAT}, we assume here that the entangled systems used in the protocol are shared in advance between the parties, becoming inaccessible to the adversary after this preliminary stage. A generalised EAT has been proposed \cite{genEAT} that allows us to relax this assumption, at the cost of lower rates of certified entropy. For this reason, we have decided to compare our results to the original EAT.

Analogously to the previous section, the initial step consists of finding a min-trade-off function $f$, in this case relative to the von Neumann entropy:
\begin{equation}\label{eq:EATmintradeoff}
	f(\mu)\leq \inf_{\sigma\in\Sigma(\mu)} H(C|Z,E)_\sigma\,,
\end{equation}
where the set $\Sigma(\mu)$ describes all quantum states $\sigma$ of the system compatible with some behaviour $\mu$.

Although one can in principle optimise $f$ for the typical  behaviour $p$  using the Brown-Fawzi-Fawzi (BFF) algorithm \cite{brown2021device}, this approach was too computationally demanding in all the cases analysed in this paper. We instead ran the algorithm to find min-trade-off functions $f_z$ for each fixed input $z$, and from there constructed the min-trade-off function $g(\mu):=\sum_z p(z)f_z(\mu)$, which satisfies
\begin{equation}
	g(\mu) \leq \sum_z p(z) \inf_{\sigma\in\Sigma(\mu)} H(C|z,E)_\sigma \leq \inf_{\sigma\in\Sigma(\mu)} \left(\sum_z p(z)H(C|z,E)_\sigma\right) = \inf_{\sigma\in\Sigma(\mu)} H(C|Z,E)_\sigma\,.
\end{equation} 
By direct application of the EAT, the extractable entropy is
\begin{equation}\label{eq:HminEAT}
	H_{\min}^{\varepsilon/p_{Acc}}(\mathbf{C}|\mathbf{Z},E,\verb|Accept|)\geq  nt-\xi\sqrt{n}\,,
\end{equation}
where $g(\mu)\geq t$ is the chosen acceptance condition and
\begin{equation}
	\xi=2(\log(1+2|C|)+\lceil \| \nabla g \|_\infty \rceil)\sqrt{1-2\log{\varepsilon}}\,,
\end{equation}
where $\| \cdot \|_\infty$ is the infinity (or maximum) norm. In our analyses, we used the threshold $t=g(p)$.

\subsection{Randomness amplification: non-signalling adversaries with classical side information for bipartite Hardy settings}

In ~\cite{ramanathan2018practical}, the authors use the number of random bits at the end of the protocol as their figure of merit, which assumes the use of a specific randomness extractor. Here we are interested in the amount of extractable entropy, which cannot be readily obtained from this reference. We thus go back to the seminal paper \cite{brandaoRealisticAmpFew} and derive an expression for the extractable entropy that can be fairly compared with our results.

We start by using Proposition~3 in Appendix I.E of \cite{ramanathan2018practical} which ensures that any NS probability distribution $p(AB|XY)$, for which a value $h$ of the MDL inequality \eqref{eq:MDL_ineq} is observed, can be bounded by
\begin{equation}
	p(A,B|X,Y)\leq 1-\frac{h}{(\frac{1}{4}-\delta^2)^2}\,.
\end{equation}
The following results allow us to relate the violation $h$ to the experimentally accessible quantities. Following \cite{ramanathan2018practical}, we define the random variable
\begin{equation}\label{eq:RV_MDL}
	M^\delta=			
	\begin{cases}
		-\left(\frac{1}{2}+\delta\right)^2&  \text{if } ABXY=0101,1010,0011\\
		\left(\frac{1}{2} -\delta\right)^2&  \text{if } ABXY=0000\\
		\qquad 0&  \text{otherwise}\,,
	\end{cases}
\end{equation}
whose expected value corresponds to $h=\mathbb{E}(M^\delta)$ whenever inputs $X$ and $Y$ are each described by a $\delta$-SV-source. An estimate for $h$ can then be computed by assuming an IID behaviour of the device, leading to a random variable
\begin{equation}\label{eq:est_MDL}
	L_\delta=\frac{1}{n}\sum_{i=1}^n M^\delta_i\,,
\end{equation} 
whose values correspond to the experimentally observed violation of the MDL inequality \eqref{eq:MDL_ineq}. However, from the point of view of an adversary knowledgeable of the memory effects of the device, the corresponding random variable should be given by
\begin{equation}\label{eq:est_MDL_Eve}
	\overline{L}_\delta=\frac{1}{n}\sum_{i=1}^n \mathbb{E}\left(M^\delta_i|\mathbf{W}_{i-1},e\right)\,,
\end{equation} 
where $\mathbf{W}_{i-1}=A_1B_1X_1Y_1\ldots A_{i-1}B_{i-1}X_{i-1}Y_{i-1}$ is information accumulated over the past rounds. 
Lemma 2 of Appendix I.B in \cite{ramanathan2018practical} uses the Azuma-Hoeffding inequality to obtain the confidence bound
\begin{equation}\label{eq:MDL_confbound}
	\Pr(\overline{L}_\delta\geq L_\delta-s_{Az})\geq 1-2e^{-n\frac{{s_{Az}}^2}{2}}:=1-\epsilon_{Az}\,,
\end{equation}
where $s_{Az}\geq 0$ is a constant.
The following lemma relates multi-round $\overline{L}_\delta$ to the single round quantity $\overline{M}_i^\delta:=\mathbb{E}\left(M^\delta_i|\mathbf{W}_{i-1},e\right)$.
\begin{lemma}[Lemma 3 in Appendix I.B of \cite{ramanathan2018practical}]\label{th:goodrounds}
 	If $\overline{L}_\delta\geq l $ then $\overline{M}_i^\delta\geq k$ in, at least, $\frac{l-k}{1/16-k}n$ rounds, for $k\in [0,l]$.
\end{lemma}
 Let us now define the acceptance threshold of our protocol to be $L_\delta\geq h_{exp}$, where $h_{exp}$ is the observed violation of the MDL inequality \eqref{eq:MDL_ineq}. In this case, $\overline{L}_\delta\geq h_{exp}-s_{Az}$ with probability given by \eqref{eq:MDL_confbound}. Moreover, using Lemma \ref{th:goodrounds}, we obtain 
 \begin{equation}
 	\mu(A_iB_i|X_iY_i,\mathbf{W}_{i-1},E, L_\delta\geq h_{exp})\leq 1-\frac{h_{exp}}{(\frac{1}{4}-\delta^2)^2}=:\gamma
 \end{equation}
in at least a fraction $\alpha(k) :=\frac{h_{exp}-s_{Az}-k}{1/16-k}$ of the rounds, for $k\in[0,h_{exp}-s_{Az}]$.

 Going back to our concise input/output notation $C=AB$ and $Z=XY$ and using Lemmas 12, 18 and Proposition~19 in Appendix B of \cite{brandaoRealisticAmpFew}, we obtain  
\begin{equation}
	\mathbb{E}\left(\max_{\mathbf{c}}\mu(\mathbf{c}|\mathbf{Z},E,\verb|Accept|)\right)\leq \frac{\max(\epsilon_{Az},\gamma^{\alpha n})+\epsilon_{Az}}{p_{Acc}}\leq\frac{\gamma^{\alpha n}+2\epsilon_{Az}}{p_{Acc}}\,.
\end{equation}
From Markov's inequality, we then derive the confidence bound: 
\begin{equation}\label{eq:confboundAmp}
	\textrm{Pr}\left(\mu(\mathbf{C}|\mathbf{Z},E,\verb|Accept|) \geq  \frac{\gamma^{\alpha n}+2\epsilon_{Az}}{\epsilon'+\epsilon_{Az}}\right)\leq \frac{\epsilon'+\epsilon_{Az}}{p_{Acc}}\,, 
\end{equation}
where $\epsilon'\geq 0$ is a free optimisation parameter. Using Theorem~\ref{th:confbound2minent}, we finally obtain the extractable entropy:
\begin{equation}
	H_{\min}^{\varepsilon/p_{Acc}}(\mathbf{C}|\mathbf{Z},E,\verb|Accept|)\geq \max_k\left(-\log\left(\gamma^{\alpha(k)n}+2\epsilon_{Az}\right)\right)+\log \varepsilon
\end{equation}
where  $\varepsilon=\epsilon_{Az}+\epsilon'$ is our total security parameter.

\end{document}